\documentclass[12pt]{article}
\usepackage[authoryear]{natbib}
\usepackage{amsmath,amssymb,amsfonts,amsthm}
\usepackage{epic,eepic,epsfig,longtable}
\usepackage{multirow,verbatim}
\usepackage{array}
\usepackage{comment}
\usepackage{float}
\usepackage{epstopdf}
\usepackage{epsfig}
\usepackage{subfigure}
\usepackage{extarrows}
\usepackage{setspace}
\usepackage{color}

 \newenvironment{tabnote}{\par\tabnotefont}{\par}
 \def\tabnotefont{\fontsize{10}{10}\selectfont}%
\makeatletter
\def\singlespace{\def\baselinestretch{1}\@normalsize}

\makeatletter
\def\singlespace{\def\baselinestretch{1}\@normalsize}

\numberwithin{equation}{section}
\renewcommand{\hat}{\widehat}

\renewcommand{\hat}{\widehat}

\newcommand{\bfm}[1]{\ensuremath{\mathbf{#1}}}

\def\ba{\bfm a}   \def\bA{\bfm A}  
\def\bb{\bfm b}   \def\bB{\bfm B}  
\def\bc{\bfm c}   \def\bC{\bfm C}  
\def\bd{\bfm d}   \def\bD{\bfm D}  
\def\be{\bfm e}     
\def\bff{\bfm f}  \def\bF{\bfm F}  
\def\bg{\bfm g}   \def\bG{\bfm G}  
   \def\bH{\bfm H}  
   \def\bI{\bfm I}

   \def\bL{\bfm L}  
   \def\bM{\bfm M}

   \def\bP{\bfm P}

   \def\bS{\bfm S}  
     
\def\bu{\bfm u}   \def\bU{\bfm U}  
\def\bv{\bfm v}   \def\bV{\bfm V}  
\def\bw{\bfm w}   \def\bW{\bfm W}  
\def\bx{\bfm x}   \def\bX{\bfm X}  
\def\by{\bfm y}   \def\bY{\bfm Y}  
\def\bz{\bfm z}   \def\bZ{\bfm Z}

\newcommand{\bfsym}[1]{\ensuremath{\boldsymbol{#1}}}

 \def\balpha{\bfsym \alpha}
 \def\bbeta{\bfsym \beta}			 \def\bPhi{\bfsym \Phi}
 \def\bgamma{\bfsym \gamma}             
 \def\bdelta{\bfsym {\delta}}

 \def\bnu{\bfsym {\nu}}
 \def\btheta{\bfsym {\theta}}           
           \def\bepsilon{\bfsym \varepsilon}
 \def\bsigma{\bfsym \sigma}             \def\bSigma{\bfsym \Sigma}
 \def\blambda {\bfsym {\lambda}}        
           
 			\def\bphi{\bfsym \phi}
 			
 \def\bxi{\bfsym {\xi}}			 \def\bXi{\bfsym {\Xi}}
 \def\bzeta{\bfsym {\zeta}}	     \def\bell{\bfsym {\ell}}	
  \def\bvarepsilon{\bfsym {\varepsilon}}

\DeclareMathOperator{\argmin}{argmin}

\DeclareMathOperator{\cov}{cov}

\DeclareMathOperator{\diag}{diag}
\DeclareMathOperator{\E}{\mathbb E}

\DeclareMathOperator{\rank}{rank}

\DeclareMathOperator{\sgn}{sgn}

\DeclareMathOperator{\Var}{Var}

\DeclareMathOperator{\tr}{tr}

\def\newpage{\vfill\eject}

\def\today{\ifcase\month\or
  January\or February\or March\or April\or May\or June\or
  July\or August\or September\or October\or November\or December\fi
  \space\number\day, \number\year}

\newdimen\biblioindent    \biblioindent=30pt

 at 8truept

\def\sgn{\mbox{sgn}}

\newcommand{\beq}{\begin{equation}}
  \newcommand{\eeq}{\end{equation}}
\newcommand{\beqn}{\begin{eqnarray}}
  \newcommand{\eeqn}{\end{eqnarray}}
\newcommand{\beqnn}{\begin{eqnarray*}}
  \newcommand{\eeqnn}{\end{eqnarray*}}

\allowdisplaybreaks
\usepackage{verbatim}
\renewcommand{\baselinestretch}{1.66}
\numberwithin{equation}{section}
\theoremstyle{plain}
\newtheorem{thm}{Theorem}[section]

\newtheorem{ass}{Assumption}[section]
\theoremstyle{definition}

\newtheorem{algo}{Algorithm}[section]

\newcounter{CondCounter}

\usepackage{xr}
\begin{document}

\title{\LARGE Recent Developments on Factor Models and its Applications in Econometric Learning}
  \author{  Jianqing Fan\thanks{Department of Operations Research and Financial Engineering, Princeton University, Princeton, NJ 08544, USA.   \texttt{jqfan@princeton.edu}. His research is supported by NSFC grants 71991470 and 71991471. }\and  Kunpeng Li\thanks{ International school of economics and management, Capital University of Economics and Business, Beijing 100070, China; email:  \texttt{kunpenglithu@126.com}  }\and  Yuan Liao\thanks{Department of  Economics, Rutgers University, 75 Hamilton St., New Brunswick, NJ 08901, USA. \texttt{yuan.liao@rutgers.edu}}}
\date{}

\maketitle

 \singlespacing

 \begin{abstract}
This paper makes a selective survey on the recent development of the factor model and its application on statistical learnings.  We focus on the perspective of the low-rank structure of factor models,
and particularly draws attentions to estimating the model from the low-rank recovery point of view. The survey mainly consists of three parts: the first part is a review on new factor estimations based on modern techniques on recovering low-rank structures of high-dimensional models. The second part discusses statistical inferences of several factor-augmented models and applications in econometric learning models. The final part summarizes new developments dealing with unbalanced panels from the matrix completion perspective.



 \end{abstract}


{\small Key words:  factor models, spiked low rank matrix, matrix completion, unbalanced panel, multiple testing, high-dimensional
 
}

 \newpage

\tableofcontents

\onehalfspacing

 

\section{Introduction}

The recent decade has witnessed a blossom of  developments on statistical learning theories and practice, embraced with the exciting progresses on large-scale optimizations and dimension reduction techniques. Factor models, as one of the central machinery on summarizing and extracting information from large scale datasets, have    received much attention  in this revolutionary   era  of data science, and many breakthrough methodologies and applications have been developed in this exciting area.

This paper makes a selective overview on the recent developments of the factor model and its applications on econometric learning.  Our review focuses on the  perspective of the low-rank structure of factor models,
and  draws particular attentions to estimating the model from the low-rank recovery point of view.
A central focus in the progress of this literature is the understanding and recovering \textit{low-rank structures} of high-dimensional models. Many  new learning theories and methods have been developed, which have revolutionized the modern understanding of econometric modeling. Meanwhile, the low-rank structure is one of the key properties of factor models. While this structure has  long  been aware of by researchers, 
studying the factor model from the perspective of low-rank matrix recovery is relatively new, and has led to many exciting  new  discoveries and understanding.

The survey mainly consists of three parts: the first part is a review on new factor estimation based on modern techniques on recovering low-rank structures of high-dimensional models. The second part discusses statistical inferences of several factor-augmented models and applications in statistical learning models. The final part summarizes new developments dealing with unbalanced panels from the matrix completion perspective.

We concentrate on recent developments on methodologies and applications in econometric learning.
For a more comprehensive account on this topic, see Chapters 9-11 of the book by \cite{fan2020statistical}.
Meanwhile, several important topics are not covered in this survey, but have also generated extensive researches in the literature. Those include selecting the number of factors, weak factors, identification, continuous-time and time-varying models, nonstationarity and
structural breaks, Bayesian methods, bootstrap factors, as well as  more sophisticated panel data models.
Several excellent reviews have been written with  emphasis on these topics. For those reviews, we refer to
\cite{stock2016dynamic} for dynamic factor models with applications on macroeconomics, to \cite{bai2016econometric} for time series  and  panel data models, and to \cite{gagliardini2019estimation} for a recent review on conditional factor models with applications to finance.  Another class of estimation is a hybrid  of  PCA-method and the state space approach, see  \cite{giannone2008nowcasting} and \cite{doz2011two} for more discussions.  In addition, the generalized dynamic factor model is another important strand of literature, where factors are often estimated using the  dynamic principal components,  the frequency domain analog of principal components,  developed by
\cite{brillinger1964frequency}. \cite{forni, forni2005generalized}   provided rates of convergence  of the common component estimated by dynamic principal components. Finally, we refer to   the following papers
for  more detailed   developments,  among others:
 \cite{bai2002determining, AH, onatski2010determining, li2017determining},  \cite{bai2012statistical, bai2016maximum},  \cite{onatski2012asymptotics, chudik2011weak},     \cite{cheng2016shrinkage, massacci2017least, gagliardini2016time, gonccalves2018bootstrapping, baltagi2017identification,barigozzi2018simultaneous},   \cite{ait2017using, chen2019five, liao2018uniform,li2019jump, pelger2019large}, \cite{su2017time}.

We  use the following notation.
For a matrix $\bA$, let $\lambda_i(\bA)$ denote the $i$ th largest singular value of $\bA$ and  use $\lambda_{\min}(\bA)$ and $\lambda_{\max}(\bA)$ to denote its smallest and largest eigenvalues. We define  the Frobenius norm $\|\bA\|_F=\sqrt{\tr(\bA'\bA)}$,  the operator norm $\|\bA\|=\sqrt{\lambda_{\max}(\bA'\bA)}$, the element-wise norm $\|\bA\|_\infty=\max_{ij}|A_{ij}|$, and the matrix $\ell_1$-norm $\|\bA\|_{\ell_1}:=\max_{i\leq N}\sum_{j=1}^N|A_{ij}|$.  In addition, define projection matrices $\bP_\bA=\bA(\bA'\bA)^{-1}\bA$ and $\bM_\bA=\bI-\bP_\bA$  when $\bA'\bA$ is   invertible.     Finally, for two (random) sequences $a_T$ and $b_T$, we write $a_T\ll b_T$ (or $b_T\gg a_T$) if $a_T=o_P(b_T)$.

\section{Spiked Incoherent Low-Rank Models}

\subsection{The model}\label{sec: genesec}

Modern high-dimensional factor models can be viewed as a type of \textit{spiked incoherent low-rank model}, a broad class of models that have drawn active research in the recent decade. A  spiked incoherent low-rank model typically refers to a   large matrix $\bSigma$ (either observable or not), having the following decomposition:
\begin{equation}\label{eq2.1}
\bSigma=\bL+\bS.
\end{equation}
Such decomposition  satisfies the following three properties:

\begin{description}

\item[(i) Low-rank.] The rank of  $\bL$ is  either bounded or grows very slowly compared to its dimensions.

\item[(ii) Spikedness.] The nonzero singular values of $\bL$ grow fast, while the largest singular value of $\bS$ is either bounded or grows much slower.

\item[(iii) Incoherence.] (also known as ``pervasiveness")  The left and right singular vectors of $\bL$, corresponding to the nonzero singular values, should have diversified elements, which means, elements of the rescaled singular vectors should be uniformly bounded.

\end{description}

The low-rank structure achieves dimension reductions: suppose the matrix $\bSigma$ is of $N\times N_1$ dimensions, while the rank of $\bL$ is $r$. Then the low-rank structure reduces the dimension from $O(NN_1) $ to $O(N+N_1)r$; the latter is the magnitude of the number of parameters in $\bL$.  Meanwhile, the spikedness helps seperate $\bL$ from $\bS$ approximately, and ensures that the large ``signals" concentrate on $\bL$, the low rank component. Finally, the incoherence, a condition that excludes matrices being low-rank and sparse simultaneously, enables us to estimate well the singular eigenvectors.

We explain these three properties using the matrix form of factor models. Consider
\begin{equation} \label{eqjf06}
y_{it}=\bb_i'\bff_t+ u_{it},\quad i\leq N, \quad t\leq T,
\end{equation}
where $\bff_t$ is a $r$-dimensional vector of factors; $\bb_i$ is the loading vector and $u_{it}$ is the idiosyncratic noise. Specifically,  (\ref{eq2.1}) applies to two decompositions of this   model.

\textbf{Factor Decomposition.} The matrix form of the  factor model gives
$$
\bY=\bM+\bU,\quad \bM:=\bB\bF',
$$
where $\bY$ and $\bU$ are $N\times T$ matrices of $y_{it}$ and $u_{it}$; $\bB$ is the $N\times r$ matrix of $\bb_i$ while $\bF$ is the $T\times r$ matrix of $\bff_t.$ Then corresponding to the notation (\ref{eq2.1}), $\bSigma=\bY$, $\bL=\bM$ and $\bS=\bU.$ In this decomposition, $\bSigma$ is observable. Apparently, $\bM$ is a low-rank matrix with rank $r$. The nonzero singular values of $\bM$, under the strong factor assumption, grows much faster than those of $\bU$, which gives rise to the spikedness property. Now let $\bxi$ be the $N\times r$ matrix whose columns are the left singular vectors of $\bM$,  and let $\bxi_i'$   denote its $i$ th row.   Then   under the assumption that the nonzero eigenvalues of $\bB'\bB$ grow fast with $N$,   for some constant $C>0$,
\begin{equation}\label{eq0.2.2}\max_{i\leq N}\|\sqrt{N}\bxi_i\|\leq C\max_{i\leq N}\|\bb_i\|,
\end{equation}
 which gives rise to the incoherent singular vectors.  The right singular vectors can be bounded similarly.  



\textbf{Covariance Decomposition.} It is also well known from the factor model (\ref{eqjf06}) that the covariance matrix of $\by_t=(y_{1t},\cdots,y_{Nt})'$, denoted by $\bSigma_y$, can be decomposed as follows:
 \begin{equation}\label{eq2.2}
\bSigma_y=\bL+\bSigma_u,\quad \bL:=\bB\cov(\bff_t)\bB',
\end{equation}
where $\bSigma_u$ denotes the covariance matrix of $\bu_t$.
The above decomposition is well known for portfolio allocations and risk managements, where the total volatility is decomposed into the systematic risk $\bL$, plus the (sparse) idiosyncratic risk $\bSigma_u$.  It also leads to the spiked incoherent low-rank model, but $\bSigma_y$ is unknown and needs to be estimated.    

\subsection{Estimation}\label{sec2.2}

There are two general approaches to estimating  model (\ref{eq2.1}):  (i) Principal Components Analysis (PCA), and (ii) low-rank regularization.  Here we present a general PCA estimation setting, and defer the discussion of  low-rank regularization to Section \ref{sec:3.2}. We shall assume rank$(\bL)=r$ to be known.

For any matrix $\bA$, let $\bA=\bU_{A}\bD_{A}\bV'_{A}$ denote the singular value decomposition (SVD) of $\bA$. Define the \textit{singular value hard thresholding} operator as
\begin{equation} \label{eqjf05}
H_R(\bA):= \bU_A\bar \bD_R\bV_A'
\end{equation}
where $\bar \bD_R$ is a diagonal matrix that keeps  the top $R$ diagonal elements of $\bD_A$ and replaces the remaining  elements by zeros. So $H_R(\bA)$ is the best rank $R$ matrix approximation to $\bA$.

Suppose  an estimator of   $\bSigma$, denoted by $\widehat\bSigma$, is available, satisfying
\begin{equation}\label{eq2.3fdaf}
\|\widehat\bSigma-\bSigma\|=O_P(\eta_N),\quad \|\widehat\bSigma-\bSigma\|_\infty=O_P(c_N)
\end{equation}
for some sequences $\eta_N$ and $ c_N$.  We use $\widehat\bSigma$ as the input matrix, which can be the sample covariance matrix or its robustfied versions \citep{fan2019robust}.
The goal is to estimate $\bL$ in (\ref{eq2.1}) and  its $N\times r$ matrix of the left singular vectors, denoted by $\bxi$ (also let $\bzeta$ denote its right singular vectors).
We  use respectively
$
\widehat\bL:=H_R(\widehat\bSigma)
$ with $R=r$, which is the rank $r$ projection of $\widehat{\bSigma}$, and
the $N\times r$  matrix $\widehat\bxi$ whose columns are the left singular vectors of $\widehat\bSigma$.   The following theorem, adapted from \cite{fan2018eigenvector}, provides deviation bounds of the estimators.
To make the paper self-contained, we also provide a simpler proof  with slightly different conditions. 

\begin{thm}\label{th1.1} Consider the general model (\ref{eq2.1}) with bounded $r:=\text{rank}(\bL)$.  Suppose that  $ \min_{2\leq i\leq r+1}| \lambda_{i-1}(\bL)-\lambda_i(\bL)|\asymp\max_{2\leq i\leq r+1}| \lambda_{i-1}(\bL)-\lambda_i(\bL)| :=g_N$ and   $\eta_N+\|\bS\|=o_P(g_N)$.
Then, under condition (\ref{eq2.3fdaf}), we have (i)
$$
 \|\widehat\bL-\bL\|=  O_P\left( \eta_N+\|\bS\|\right),\quad \|\widehat\bxi-\bxi\| =O_P\left(\frac{ \eta_N+\|\bS\|}{  g_N}\right).
$$
 (ii)  If  additionally,     $\|\bS\|_\infty+\|\bL\|_\infty = O_P(1) $, $N_1c_N=o_P(g_N)$, then
\begin{eqnarray*}
& & \|\widehat\bxi-\bxi\|_\infty
\leq O_P  \left(   \frac{N_1}{\sqrt{N}}+  \sqrt{N_1}   \right)    (\eta_N+\|\bS\|)g_N^{-2}\\
& & \qquad
+O_P \left(c_N \frac{ N_1}{\sqrt{N}}   +  c_N  \sqrt{N_1}   +\|\bS\bzeta_d \|_\infty\vee\|\bS'\bxi_d \|_\infty\right) g_N^{-1}.
\end{eqnarray*}
\end{thm}

\begin{proof}
See the online supplement.
\end{proof}

This theorem is relatively general, and is applicable to low-rank models that
are not necessarily consequences from factor models. The proof relies on perturbation bounds for singular vectors/values, and the achieved rates are sharp.  Result (i) is simple and gives asymptotic bounds under the operator norm.
Result (ii) gives element-wise deviation bound for the singular vectors, which requires more dedicated technical arguments.    



\section{Estimation under Factor Models}

We  observe an  $N\times T$ data matrix $\bY$, which can be decomposed as
\[\bY=\bM+\bU=\bB\bF'+\bU\]
where $\bB$ is $N\times r$ factor loadings matrix, $\bF$ is $T\times r$ factors matrix and $\bU$ is $N\times T$ idiosyncratic errors, which are uncorrelated with $\bM:=\bB \bF'$. All the three parts $\bB$, $\bF$ and $\bU$ are unobserved.    The $t$ th column of this expression can be written as
\begin{equation} \label{eq3.1}
  \by_t=\bB\bff_t+\bu_t. 
\end{equation}

\subsection{PCA and MLE}\label{sec3.1}
\subsubsection{PCA}
Under the model's specification, we have the covariance structure (\ref{eq2.2}).
One of the most widely used estimation methods for the factor model is \textit{principal components analysis} (PCA).  Define the sample covariance
$\bS_y=\frac1T\sum_{t=1}^T \by_t\by_t'=\frac1T\bY\bY'$.
Let $\widehat\bxi_j$  be the $j$th eigenvector corresponding to the largest $j$ th eigenvalues of $\bS_y$.  The PCA  estimates $\bB$ by taking $\widehat\bB=\sqrt{N}(\widehat\bxi_1,\cdots,\widehat\bxi_R)$, which estimates $\bB$ up to a diagonal transformation.
   Given $\widehat\bB$, the factors can be estimated via the least squares:
   $$
   \widehat\bF=\bY'\widehat\bB (\widehat\bB'\widehat\bB)^{-1} =\frac{1}{N}\bY'\widehat\bB.
   $$
   This also leads to the estimated low-rank component
$\frac{1}{T}\widehat\bB\widehat\bF'\widehat\bF\widehat\bB'$ for $\bB\cov(\bff_t)\bB'  $.

PCA  is equivalent to the singular value hard thresholding by taking the input matrix $\widehat\bSigma=\bS_y$. Then
$\frac{1}{T}\widehat\bB\widehat\bF'\widehat\bF\widehat\bB'=H_R(\bS_y)$.
One can then apply Theorem \ref{th1.1} to infer the rates of convergence of the PCA estimators, which were obtained by \cite{SW02}.
\cite{bai03} proved the asymptotic normality of PCA estimators for the factors and loadings.
Results with general input $\hat{\bSigma}$ can be found in Chapter 10 of \cite{fan2020statistical}.

 \subsubsection{Maximum Likelihood Estimations}
Another popular method to estimate a factor model is the maximum likelihood (ML) method (see, e.g., \cite{Lawley}, \cite{bai2012statistical}, \cite{doz}). 
Under the independence and normality assumptions, the log-likelihood function based on $\by_t$ is, for some constant $C$,
\[
    \log L_{\mathrm{ML}}(\bB, \cov(\bff_t),  \diag(\bSigma_u))=C-\frac T2\ln|\bSigma_y|-\frac12\sum_{t=1}^T\by_t'\bSigma_y^{-1}
    \by_t.
\]
The log-likelihood function is then maximized with respect to the matrix parameters $(\bB, \cov(\bff_t),  \diag(\bSigma_u))$ under additional restrictions that $\bSigma_u$ is diagonal \citep{bai2012statistical, bai2016maximum} or sparse with regularizations \citep{bai2016efficient, wang2019penalized}. Recently \cite{barigozzi2019quasi} explicitly accounted for  autocorrelations of the factors   in the likelihood function.


The factors can be estimated by two methods, one of which is the projection method. Under the joint normality assumptions of $\bff_t$ and $\bu_t$, we have
\[\mathbb E(\bff_t|\by_t)=\bB'(\bB\bB'+\bSigma_u)^{-1}\by_t=(\bI_r+\bB'\bSigma_u^{-1}\bB)^{-1}\bB'\bSigma_u^{-1}\by_t.\]
This provides the basis of estimating factors. The other approach is the generalized least squares: for given $\bB$ and $\bSigma_u^{-1}$, 
the GLS estimator for $\bff_t$ is
\[\widehat\bff_t= (\bB'\bSigma_u^{-1}\bB)^{-1}\bB'\bSigma_u^{-1}\by_t.\]
Replacing the unknown parameters with their ML estimators, one obtains  two estimators for the latent factors. Under large-$N$ setup, the difference of the two methods (PCA and MLE) for estimating factors are asymptotically negligible.

\subsection{Low rank estimation}\label{sec:3.2}

 Alternative to PCA, one can estimate $\bM$ directly taking advantage of its low-rank structure, based on the \textit{nuclear-norm regularization}, the $\ell_1$-norm of singular values, that encourages the sparseness in singular values and hence low-rankness.
  For an $n\times m$ matrix $\bA$, let $\|\bA\|_n:=\sum_{i=1}^{\min\{m,n\}}\psi_i(\bA)$ be its nuclear-norm, where $\psi_i(\bA)$ is the $i$ th largest singular value of $\bA$.

\subsubsection{Singular value thresholding}

 Given the low-rank structure of $\bM$ (sparsity in singular value of $\bM$), we can estimate the model via solving the following penalized optimization:
\begin{equation}\label{soft-2.1}
\widehat \bM= \arg\min_{ M} \frac{1}{2} \|\bY- \bM\|_F^2+\nu\|\bM\|_n
\end{equation}
for some tuning parameter $\nu>0$. The solution is  $\widehat \bM= S_{\nu}(\bY)$, where $S_{\nu}(\cdot)$ is the singular value thresholding operator \citep{ma2011fixed}, defined as follows.   Let $\bY=\bU_y\bD\bV_y'$ be its  SVD. Then  $S_{\nu}(\bY):= \bU_y\bD_{\nu}\bV_y',$ where $\bD_{\nu}=
\diag(\{D_{ii}-\nu\}_+)$  with $D_{ii}$ being the  diagonal entries of $\bD$.  So $S_{\nu}(\bY)$ applies ``soft-thresholding" on the singular values of $\bY$. One   can additionally estimate the factors and loadings using the singular vectors.


We note that this method is closely related to the PC-estimator, except the soft-thresholding is replaced by hard-threshoding. Let $R$ denote the ``working number of factors", which is the number of principal components one takes when applying the PC-method. 
We note that   the PC-estimator for $\bM$ with $R$ factors is given by (see Section \ref{sec2.2}):
$$
\widehat \bM_{\text{PC}} = H_R(\bY),\quad H_R(\bY):=\bU_y\bar \bD_{R} \bV_y'.
$$
This estimator is the solution to the penalized least squares problem (\ref{soft-2.1}) except that the nuclear norm is replaced by $\sum_{i=1}^{\min\{N,T\}} p_\nu(\psi_i(\bM))$, where $p_\nu(\theta) = \nu^2 - (\nu - |\theta|)_+^2$ is the harding thresholding penalty  and
$\psi_i(\bM)$ is the $i^{th}$ singular value of $\bM$.

Therefore the difference between (\ref{soft-2.1}) and PCA is more fundamentally about that of hard- and soft- thresholding.
Despite of many good properties,    the soft-thresholding estimator  possesses  shrinkage bias, while the hard-thresholding reduces the bias.  
As a matter of fact,   the shrinkage bias is  on the singular values, rather than on the singular vectors. Indeed, the singular vectors of the two estimators are the same, and equal to the top $R$ singular vectors of $\bY.$  An important implication is that the factor estimator building on  $\widehat\bM$  is numerically equivalent to the PC-estimators for the factors, which does not suffer from any shrinkage bias.  A formal statement and proof of the unbiasedness of eigenvectors can be found in
\cite{fan2019distributed}.

\subsubsection{Low-rank plus sparse decomposition}

Recall that  $\bSigma_y$ and $\bSigma_u$ denote the $N\times N$ covariance matrices of $\by_t$ and $\bu_t$ in model (\ref{eq3.1}), and that we have the following decomposition
\begin{equation}
\bSigma_y=\bL+\bSigma_u,\quad \bL:=\bB\cov(\bff_t)\bB'.
\end{equation}
We now demonstrate  that this decomposition also provides a nice structure for estimating the covariance components.  A key assumption is \textit{conditionally sparsity}, namely, $\bSigma_u$ is sparse. While the definition of sparsity may differ in different contexts, here  we  mean
$$
J:=\sum_{i\neq j} 1\{\E u_{it} u_{jt}\}
$$
 should not grow too fast as $N\to\infty.$ This requirement can be weakened to \textit{approximate sparsity}. In addition, $\bL$ is a low-rank matrix. Thus we can directly estimate the above covariance decomposition via solving the following penalized optimization:
 \begin{equation}\label{eq3.3add}
 (\widehat\bL, \widehat\bSigma_u):=\arg\min_{\bL,\bSigma_u} \frac{1}{2} \|\bS_y- (\bL+\bSigma_u)\|_F^2+\nu_1\|\bL\|_n+\nu_2\|\bSigma_u\|_1,
 \end{equation}
where $\nu_1$ and $\nu_2$ are tuning parameters. Note that here we use the notation $\|\bA\|_1=\sum_{i,j}|A_{ij}|$ as the matrix 1-norm, distinguished from the usual matrix $\ell_1$-norm $\|\bA\|_{\ell_1}:=\max_{i\leq N}\sum_{j=1}^N|A_{ij}|$.  The above optimization has been employed by many authors to study the \textit{low rank plus sparse decomposition}, while some authors exclude the diagonal elements of $\bSigma_u$ from the penalization, and additionally impose positive-definite and other constraints on $\bL$ and $\bSigma_u$ \citep{klopp2017robust, agarwal2012noisy}. Finally, given $\widehat\bL$, we can   estimate the factors and loadings by extracting its  eigenvectors.

The above optimization can be solved by alternating the estimation of $\bL$ and $\bSigma_u$, and closed form solutions are available in both iterations.
Given $\bSigma_u$, solving for $\bL$ leads to the \textit{singular value soft-thresholding}:  $\widehat{\bL} = S_{\nu_1}(\bS_y- \bSigma_u)$, and
given $\bL$, solving for $\bSigma_u$ leads to the \textit{element-wise soft-thresholding}: $\widehat{\bSigma}_u =  \widetilde S_{\nu_2}(\bS_y-\bL)$.
While both  iterations solve convex problems, standard convergence analysis can be applied to show that the iterative algorithm converges in polynomial time.

\cite{agarwal2012noisy} and \cite{klopp2017robust} studied the statistical convergence properties of (\ref{eq3.3add}). Let columns of $\bU_{L,2}$ be  the singular vectors of the true $\bL$ corresponding to the \textit{zero}  singular values. Define projections $\mathcal P(\bA):=\bU_{L,2}\bU_{L,2}'\bA\bU_{L,2}\bU_{L,2}'$ and $\mathcal M(\bA):=\bA-\mathcal P(\bA)$. In addition, let $(\bA)_J$ and $(\bA)_{J^c}$ be the submatrices of $\bA$, whose elements respectively correspond to $\E u_{it}u_{jt}\neq0$ and  $\E u_{it}u_{jt}=0.$  Additionally define
$$
\mathcal C(\nu_1, \nu_2):=\{(\bA_1,\bA_2):
\nu_1\|\mathcal P(\bA_1)\|_n+\nu_2\|(\bA_2)_{J^c}\|_1\leq 3\nu_1\|\mathcal M(\bA_1)\|_n+3\nu_2\|(\bA_2)_{J}\|_1
\}.
$$
A key quantity is the \textit{restricted strong convexity} (RSC) constant, which is defined as follows:
$$
\kappa(\nu_1,\nu_2):=\sup\{c>0:\|\bA_1+\bA_2\|_F^2\geq c\|\bA_1\|_F^2+c\|\bA_2\|_F^2 \text{ for all } (\bA_1, \bA_2)\in \mathcal C(\nu_1, \nu_2) \}.
$$

We then have the following theorem, adapted from \cite{agarwal2012noisy}. To make the paper self-contained, we also provide a   proof  with slightly different conditions. See the online supplement.

\begin{thm}\label{th2.1}
Conditioning on  events
 $4\|\bS_y-\bSigma_y\|\leq \nu_1$ and $ 4\|\bS_y-\bSigma_y\|_\infty\leq \nu_2$, there is $C>0$ that only depends on $\text{rank}(\bL)$, so that
 $$
 \frac{1}{N^2}\|\widehat\bL-\bL\|_F^2+ \frac{1}{N^2}\|\widehat\bSigma_u-\bSigma_u\|_F^2\leq \frac{C}{\kappa^2(\nu_1,\nu_2)}\frac{(\nu_1^2+(J+N)\nu_2^2)}{N^2}.
 $$
\end{thm}

\begin{proof}
See the online supplement.
\end{proof}

The optimal  tuning parameters can be set to satisfy $\nu_1\asymp \frac{N}{\sqrt{T}}$ and $\nu_2\asymp \sqrt{\frac{\log N}{T}}$, respectively, accounting for estimating errors under two matrix norms:
 $$
 \|\bS_y-\bSigma_y\|\leq \nu_1,\quad  \|\bS_y-\bSigma_y\|_\infty\leq \nu_2;
$$
both can be shown to hold with high probability under weak serial dependence and sub-Gaussian conditions. In additionally, if $\kappa(\nu_1,\nu_2)$ is bounded away from zero, with the choice of tunings,  the convergence rate in Theorem \ref{th2.1} is $O_P(1+\frac{J\log N}{N^2})\frac{1}{T}$, which is sufficient to guarantee the convergence of the estimated factors and loadings. We refer to  Lemma 2 of \cite{agarwal2012noisy}  for more refined lower bound of $\kappa(\nu_1,\nu_2)$.

The above problem is also called ``robust PCA'' \citep{candes2011robust}.  For recent advance and references, see \cite{chen2020bridging} where factorization methods are also discussed. 


\subsection{Covariance estimation}\label{sec:poet}

 \cite{POET} proposed a nonparametric estimator of $\bSigma_y$, named POET (Principal Orthogonal complEment Thresholding), when the factors are unobservable. It is basically an one-step solution to optimization (\ref{eq3.3add}) with initialization $\bSigma_u=0$.  To motivate the estimator, suppose  $r=R$. Then, heuristically
 $$
\bL\approx H_R(\bSigma_y),\quad \bSigma_u\approx  \bSigma_y -H_R(\bSigma_y),
 $$
 Thus, one estimates $\bL$ by $H_R(\bS_y)$ and sets $\bS_u:=\bS_y-H_R(\bS_y)$.
 To account for the sparsity assumption on $\bSigma_u$,  \cite{POET}  estimates $\bSigma_y$ and $\bSigma_u$ as
 \begin{equation}
 \widehat\bSigma_y=H_R(\bS_y)+\widehat\bSigma_u,\quad \widehat\bSigma_u= (h(S_{u,ij},\lambda_{ij}))_{N\times N},
 \end{equation}
 where  $h(x,\lambda_{ij})$  denotes the element-wise thresholding operator with thresholding value $\lambda_{ij}$. Here, we emphasize element-dependent thresholding $\lambda_{ij}$ to adapt to varying scales of covariance.  For correlation thresholding at level $\lambda$, we take $\lambda_{ij} = \lambda \sqrt{s_{u, ii} s_{u, jj}}$ with $s_{u, ii}$  a diagnonal element of $\bS_u$\citep{POET}; we can also take other form such as the adaptive thresholding in \cite{Cai11b}.  In general, the thresholding function should satisfy:
\\
(i) $h(x,\lambda)=0$ if $|x|<\lambda$,\\
(ii) $|h(x,\lambda)-x|\leq \lambda$.\\
(iii) there are constants $a>0$ and $b>1$ such that $|h(x,\lambda)-x|\leq a\lambda^2$ if $|x|>b\lambda$.

Note that condition (iii) requires that the thresholding bias should be of higher order.  It  is not necessary for consistent estimations, but we recommend using   nearly unbiased thresholding  \citep{AF} for  inference applications. One such example is known as SCAD.
As noted  in  \cite{powerenhancement},   the unbiased thresholding is required  to avoid   size distortions in a large class of high-dimensional testing problems involving a ``plug-in" estimator of $\bSigma_u$.
In particular, this rules out the popular  {soft-thresholding} function,  which does not satisfy (iii) due to its first-order shrinkage bias.  




\subsection{Projected PCA}

In empirical asset pricing, factor loadings are known to  depend on   individual-specific observables $\bX_i$, which represent a set of time-invariant  characteristics such as individual stocks' size, momentum, and values.
To incorporate the information carried by the observed characteristics,   \cite{CL07} and \cite{CMO} model explicitly the loading matrix as a function of covariates $\bX_i$.  \cite{fan2016projected}  extended the model to allowing components in factor loadings that are not explainable by characteristics:
\begin{equation}\label{eq4}
\bb_i=\bg(\bX_i) +\bgamma_i,\quad \E(\bgamma_i|\bX_i)=0.
\end{equation}
Here $\bg(\cdot)$ is a vector of nonparametric functions. With this model, they introduced an improved factor estimator, known as \textit{projected PCA}.

The basic idea of projected PCA is to smooth the observations $\{y_{it}\}_{i=1}^N$ for each given  $t$ against their associated covariates $\{\bX_{i}\}_{i=1}^N$ (cross-sectional smoothing), and apply   PCA to the smoothed data (fitted values). Let $\{\phi_j(\bx)\}_{j=1}^J$ be a set of basis functions. This can be either unstructured, such as kernel machines, or structured such as a basis for additive models \citep{fan2020statistical}. Set $\phi(\bX_i)'=(\phi_1(\bX_{i}),\cdots.,\phi_J(\bX_{i}))$ and $\Phi(\bX)=(\phi(\bX_1),\cdots,\phi(\bX_N))'$, an $N\times J$ matrix. Then the projection matrix on  characteristics can be  taken as
$
\bP=\Phi(\bX)(\Phi(\bX)'\Phi(\bX))^{-1}\Phi(\bX)'.
$
 The projected data $\bP \bY$ is the fitted value of regressing $\bY$ on to the   basis functions.

We make the following key assumptions:

\begin{ass}
\label{a3}
\begin{description}

\item[(i) Relevance:] With probability approaching one, all the eigenvalues of $\frac{1}{N}({\bP}\bB)'{\bP}\bB$ are bounded away from both zero and infinity as $N\to\infty$.

\item[(ii) Orthogonality:] $\mathbb{E}(u_{it}|\bX_i)=0$ for all $i\leq N, t\leq T.$

\end{description}
\end{ass}

The above  conditions require that the strengths of the  loading matrix should remain strong after the projection.  Condition (ii) implies that if we apply $\bP$ to both sides of $\bY=\bB\bF'+\bU$, then
$$
\bP\bY\approx\bP\bB\bF' = \bG \bF'
$$
where $\bG=\bP\bB$ is the $N\times r$ matrix, which $\approx (\bg(\bX_i))_{N\times r}$ under additional assumption $ \E(\bgamma_i|\bX_i)=0$ for all $i\leq N$.  In other words, the noise $\bU$ is suppressed, while signals remain.  Hence, the scaled sample covariance  $(\bP\bY)'\bP\bY =
\bY'\bP\bY\approx\bF\bG'\bG\bF'.
$
 For identification purpose, let us assume   $\bXi:=\bG'\bG$ is a diagonal matrix and $\bF'\bF/T=\bI$. Then from
$$
\frac{1}{T} \bY'\bP\bY\bF\approx  \bF\bXi,
$$ we infer that the columns of $\bF$ are approximately  the eigenvectors of the $ \bY'\bP\bY$, scaled by a factor $\sqrt{T}$.  This motivates estimating factors by using the top $R$ eigenvectors of $ \bY'\bP\bY$.

 \cite{fan2016projected}  derived the rates of convergence of the projected PCA method.   A nice feature  is that the consistency of latent factors is achieved even when the sample size $T$ is finite so long as $N$ goes to infinity.  Intuitively, the  idiosyncratic noise is removed from cross-sectional projections, which does not require a long time series.

Similarly, in many applications, while we do not know the latent factors $\bff_t$, we do know that factors are related to some proxy variables $\bW_t$.  For example, the latent factors are unknown for equity markets, but they are related to Fama-French factors \citep{fama2015five}; latent factors for disaggregated macroeconomics time series are unknown, but they are related to aggregated ones \citep{mccracken2016fred}.
Switching the roles of rows and columns, longitudinal regression of each series $\{y_{it}\}_{t=1}^T$ on $\{\bW_t\}_{t=1}^T$ yields the projected data matrix, from which latent factors and loadings can be extracted similarly.  See \cite{fan2020augmented} for details on how latent factor learning is augmented by instruments $\{\bW_t\}_{t=1}^T$.

\subsection{Diversified projection}

In this section, we   continue denoting by $R$ as the number of factors we use, and by $r$ as the true number of factors.   \cite{fan2019learning} proposed a simpler factor estimator that does not rely on eigenvectors, by using cross-sectional \textit{diversified projections} (DP).   Let
$
\bW=(\bw_1,\cdots,\bw_R)
$
be a  given  exogenous (or deterministic) $N\times R$ matrix,    where each of its $R$ columns $\bw_k$   is an $N\times 1$ vector of ``diversified weights", whose definition is to be clear below. We  estimate $\bff_t$ by simply taking
$$
\widehat\bff_t=\frac{1}{N}\bW'\by_t.
$$ By substituting  $\by_t=\bB\bff_t+\bu_t$ into the definition,  immediately we have  \begin{equation}\label{eq1.2}
\widehat\bff_t=\bH\bff_t+\frac{1}{N}\bW'\bu_t,\quad \bH=\frac{1}{N}\bW'\bB.
\end{equation}
Thus $\widehat\bff_t$ (consistently) estimates $\bff_t$ up to an $R\times r$ affine transform $\bH$, with the  estimation error $\be_t:=\frac{1}{N}\bW'\bu_t$. The assumption that $\bW$ should be diversified ensures that as $N\to\infty$, $\be_t$ is ``diversified away" (converging to zero in probability). More specifically, we impose the following assumption.
\begin{ass}\label{ass2.1}  There is a  constant  $c>0$, so that  as $N\to\infty$,\\
	(i) The $R\times R$ matrix $\frac{1}{N}\bW'\bW$ satisfies $\lambda_{\min}(\frac{1}{N}\bW'\bW)>c.$\\
	(ii) $\bW$ is independent of $\{\bu_t: t\leq T\}$.\\
	(iii) Suppose $R\geq r$, $\rank(\bH)=r$ and
	$\psi^2_{\min}(\bH)\gg \frac{1}{N}  $, where $\psi_{\min}(\bH)$ denotes the minimum nonzero singular value of $\bH =\frac{1}{N}\bW'\bB$. 	
	\end{ass}

 Conditions (i) and (ii) define the ``diversified weights" $\bW$.
 When $(u_{1t},\cdots,u_{Nt})$ are cross-sectionally weakly dependent, they ensure that  $ \be_t$ is diversified away.  Condition (iii) of Assumption \ref{ass2.1} is  a  key condition, which requires that $\bW$  should not diversify away the factor components in the time series.    Several choices of $\bW$ can be recommended to satisfy this condition.  For instance, if  factor loadings satisfy (\ref{eq4}), then fix $R$  components of sieve basis functions: $(\phi_1(\cdot),\cdots,\phi_R(\cdot))$, we can define
      $$\bW:=  (w_{i,k})_{N\times R},\quad \text{ where } w_{i,k}=\phi_k(\bX_i).$$
Alternatively, we can also use   transformations of the initial observation $\bx_{t}$ for $t=0$, which was considered by \cite{juodis2020linear}.  If $\by_0$ is independent of $\{\bu_t: t\geq 1\}$,  we can apply
 $
 w_{i,k} = \phi_k(y_{i,0})$ .
  These weights are correlated with $\bB$ through $\by_0=\bB\bff_0+\bu_0$.


An important benefit of the  DP is that it is robust to over-estimating the number of factors.
Theoretical  studies of factor models have  been crucially depending on the assumption that the number of factors, $r$, should be consistently estimated. This usually requires  strong conditions on the strength of factors and serial conditions. 
Recently,   \cite{barigozzi2018consistent}  proposed a PCA-based method to estimate factors that are robust to over-estimated $r$.  They provided rates of convergence of the estimated common components when $R\geq r$.


   \cite{fan2019learning} applied DP to several  inference problems in   factor-augmented models, including the  post-selection inference,  high-dimensional covariance estimation, and factor specification tests. They formally justified the  robustness  to over-estimating the number of factors in these applications. In particular,  DP admits   $r=0$ but  $R\geq 1$ as a special case. That is, the inference is still valid  even if there are  no common factors present,  but   factors are nevertheless estimated for insurance. In addition, \cite{karabiyik2019cce} applied DP to the context of panel data models in the presence of common factors.

\subsection{Factor estimators robust to heavy tails}\label{sec: 3.6}

To apply  either the PCA or the MLE to estimate the model, we need an initial covariance estimator $\bS_y$,  whose application requires elements of $\by_t$ have  sufficient moments. Some technical results  of factor estimations even require sub-Gaussian conditions on data's tail distributions.  However, heavy tailed data are not uncommon in economic applications. For instance,  about thirty percent of 131 disaggregated macroeconomic variables of
\cite{ludvigson2010factor}  have excess     kurtosis greater than six, so  their distributions are fatter than the t-distribution with degrees of freedom five.
Indeed, heavy tails are a stylized feature of high-dimensional data, as it is unlikely that all variables have sub-Gaussian tails.

Because the  presence of heavy-tailed data invalidates many conditions required for estimating  factor models, the recent literature has proposed several   methods that are robust to the tail distributions. Here we describe two of them: truncation and robust M-estimation.


 


In the high-dimensional setting,    consider estimating  multivariate   means from  an independent triangular array   variables $y_{i1},...,y_{iT}$ with $\max_{i\leq N}\Var(y_{it})\leq\sigma^2$. Truncate the data
$$
\widetilde y_{it}:=\sgn(y_{it})\min\{|y_{it}|, \tau_i\}
$$
with predetermined $\tau_i>0$. We then estimate
$\mathbb Ey_{it}$ using the truncated-mean
$\widetilde y_i:=\frac{1}{T}\sum_{t=1}^T\widetilde y_{it}$. Theorem \ref{th2.2} shows that the  high-dimensional means can be estimated  uniformly   well if $\mathbb E|y_{it}|^q<M$   for some $q\geq 2$.

\cite{catoni2012challenging} constructed a robust M-estimator that shares the same Gaussian concentration.  \cite{fan2017estimation, fan2019robust} used the adaptive Huber's loss to define the mean estimator:  
 $$
 \widehat y_i=\arg\min_{\mu} \sum_{t=1}^T\psi_{\tau_i}(y_{it}-\mu)
 $$
 where $\tau_i$ is a growing sequence, and $$
 \psi_{\tau}(z)=\begin{cases} z^2\tau^{-2}, & |z|<\tau\\
 2|z|\tau^{-1}-1, & |z|\geq \tau.
 \end{cases}
 $$
 The following theorem shows that $ \widehat y_i$ also estimates $\mathbb E y_{it}$ well provided that $\max_i\mathbb E y_{it}^2$ is bounded.

 \begin{thm}\label{th2.2}  Suppose $y_{it}$ is i.i.d. across $t$,    and $\max_{i\leq N}\mathbb E y_{it}^2<\sigma^2$.

 	(i) The truncation approach:
Suppose    $\max_{i\leq N}\mathbb E|y_{it}|^q<M$   for some $q\geq 2$.  In addition, suppose $\log N\leq CT$ for some $C>0$, and the   truncation parameter is set to satisfy
$\tau_i\asymp
\left(\frac{T}{\log N}\right)^{1/(1+q/2)}(\sigma^2\max_{i\leq N}\mathbb E|y_{it}|^q)^{1/(2+q)}.
$
Then there is $c>0$ which does not depend on any moments of $y_{it}$, or $(N,T)$, with probability at least $1-2N^{-3}$,
$$
\max_{i\leq N}\left|\widetilde y_i- \mathbb E y_{it}\right|
\leq (c M^{1/(2+q)}+3)\sigma\sqrt{\frac{\log N}{T}}.
$$

(ii) The robust M-estimation approach:  Suppose $\log N=o(T)$, and the   truncation parameter is set to satisfy
$\tau_i\asymp
\sqrt{\frac{T}{\log N} }.
$
Then there is $c>0$ which does not depend on any moments of $y_{it}$, or $(N,T)$, with probability at least $1-4N^{-3}$,
$$
\max_{i\leq N}\left|\widehat y_i- \mathbb E y_{it}\right|
\leq C(\sigma+1) \sqrt{\frac{\log N}{T}}.
$$

 	\end{thm}

\begin{proof}
See appendix.
\end{proof}

The robust mean estimation also applies to estimating covariance as its $(i,j)$ element is of form $\mathbb E y_{it} y_{jt}$. When the high-dimensional data    have heavy-tailed components,	
we can replace the sample covariance  by its robust version $\widehat \bS_y$ before estimating the factors.
By the Gaussian concentration inequality, the robustly estimated covariance $\widehat \bS_y$ satisfies
$$
\| \widehat \bS_{y}- \bSigma_y\|_\infty=O_P\left(\sqrt{\frac{\log N}{T}}\right),
$$
so long as $\E y_{it}^2y_{jt}^2$ is uniformly bounded (and serial independence is assumed). 

Based on the above robust covariance inputs,  we can create factor    estimators and derive their theoretical properties following the guidance of Section \ref{sec2.2}.   See Chapter 10 of \cite{fan2020statistical} for further generalizations.



\subsection{Use of cross-covariance}

When factors are highly persistent but $\mathbb E \bu_t \bu_{t-h}^T = 0$, then the cross-covariance
$$
 \bSigma_h = \mathbb E \by_t \by_{t-h}' = \bB (\mathbb E \bff_t \bff_{t-h}) \bB', \qquad h \geq 1
$$
contains valuable  information about $\bB$. This motivates to estimate loadings by  applying PCA to  aggregated $\{\bSigma_h: h=1,\cdots\}$, and we studied by    \cite{lam2012factor}.  A related  idea has been extended to matrix-variate PCA \citep{wang2019factor,chen2020constrained}. \cite{fan2018optimal} also provided a procedure to efficiently aggregate the cross-covariance information with the covariance information when $h=0$.


\subsection{Which method to use?}

Many references have documented the comparisons among various estimation methods.   \cite{westerlund2013estimation} made a comparison between PCA and cross-sectional averages in the panel data setting.   Meanwhile,  the PCA and low-rank penalized regressions are practically very similar. So   we do not distinguish their use in practice.  In general, because of the simplicity for implementations and relatively weak required conditions, the PCA still seems to be the most widely used method in applied research.  Meanwhile,   robust covariance inputs can also be integrated with the surveyed low-rank recovery methods. 

In addition, when either factors or loadings can be partially explained by observed characteristics, the projected PCA is recommended. 
This is particularly useful in asset pricing applications where the explanatory power of asset characteristics  has been  well documented in the literature. 

\section{Factor-Augmented Inference and Econometric Learning}

\subsection{Forecasts}

Forecasting in a data-rich environment has been an important research topic in economics and finance. Typical examples include forecasts of the aggregate output or inflation rate using a large number of the categorized macroeconomic variables. 

\cite{SW02, bai2006confidence} considered   factor-augmented regression model for $h$-step ahead forecast:
 \begin{align}
y_{t+h}&=\balpha'\bff_t+\bbeta'\bw_t+\varepsilon_t\label{eqli01}\\
\bx_{t}&=\bB\bff_t+\bu_{t}.\label{eqli02}\end{align}
Here  $\bw_t$ in (\ref{eqli01}) is the observed predictors, which  may include lagged dependent variables. Equation (\ref{eqli02}) is a high-dimensional    factor model that includes a vector of latent factors $\bff_t$. The forecast can be implemented by  regressing  $y_{t+h}$ onto $\bw_t$ and estimated factors.  The factor model (\ref{eqli02}) serves as an important dimension reduction tool.


\subsubsection{Inverse regression}
\cite{fan2015sufficient} generalized  (\ref{eqli01}) to the nonlinear model with multi-indices. Consider the following forecasting model:
\begin{align}
 y_{t+1}&=h(\bphi_1'\bff_t, \dots, \bphi_R'\bff_t,\bepsilon_{t+1})\label{eqli05}
\end{align}
where   $h(\cdot)$ is an unknown link function, and $\bepsilon_{t+1}$ is the error independent of $\bff_t$ and $\bu_{t}$.   Vectors   $\bphi_1, \dots, \bphi_R$ are $r$-dimensional linear-indepencent prediction indices. In contrast with linear forecasting, the above model specifies that the predicting function is nonlinear and depends on multiple indices of extracted factors. If we specify  $R<r$, further dimension reductions are achieved.

A prominent result related to  model (\ref{eqli05}) is given by \cite{li1991sliced}, which shows that under some regularity conditions such as $\bff_t$ is elliptically symmetric, we have
\begin{equation}\label{eqjf02}
  \mathbb   E(\bff_t| y_{t+1})=\bPhi \ba( y_{t+1}),
\end{equation}
for  a $R$-dimensional vector $\ba( y_{t+1})$,
where $\bPhi=[\bphi_1, \bphi_2, \dots, \bphi_R]$ is an $r\times R$ matrix. In other words, the ``inverse regression vector" $  \mathbb   E(\bff_t| y_{t+1})$ falls in the column space spanned by $\bPhi$, which can be extracted by PCA.  Indeed, since $ \mathbb   E( \mathbb   E(\bff_t| y_{t+1}))= \mathbb   E(\bff_t)=0$,
\[\cov( \mathbb   E(\bff_t| y_{t+1}))=\bPhi  \mathbb   E[\ba( y_{t+1})\ba( y_{t+1})']\bPhi'\]
The above matrix has $R$ nonvanishing eigenvalues if $ \mathbb   E[\ba( y_{t+1})\ba( y_{t+1})']$ is non-degenerate. Their corresponding eigenvectors have the same linear span as $\bphi_1, \dots, \bphi_R$ do. If one can consistently estimate $\cov( \mathbb   E(\bff_t| y_{t+1}))$, then the subspace spanned by $\bphi_1, \dots, \bphi_R$, which is of our primary interests, can be obtained by extracting the top $R$ eigenvectors of the estimated covariance matrix that correspond to the $R$ largest eigenvalues.

However, it is not an easy task to directly estimate the covariance of $ \mathbb   E(\bff_t| y_{t+1})$.  \cite{li1991sliced} suggested the \textit{sliced covariance estimate},   a widely used technique for dimension reductions:  The sliced covariance matrix also satisfies the fundamental property (\ref{eqjf02}), namely $E(\bff_t | y_{t+1} \in \bI_k)$ falls in the column space spanned by $\bPhi$ for any given partition of the range of $ y_{t+1}$ into $H$ ``slices" $\bI_1, \bI_2, \dots, \bI_H$.  Correspondingly, let
\begin{align}\widehat{\cov( \mathbb   E(\bff_t| y_{t+1}))}&=\frac1H\sum_{h=1}^H\bigg[\frac1{\sum_{t=1}^T1( y_{t+1}\in\bI_h)} \sum_{t=1}^T\bff_t1( y_{t+1}\in\bI_h)\bigg]\nonumber\\
&\quad\qquad\quad\times\bigg[\frac1{\sum_{t=1}^T1( y_{t+1}\in\bI_h)} \sum_{t=1}^T\bff_t1( y_{t+1}\in\bI_h)\bigg]',\label{eqli007}\end{align}
which is a nonparametric covariance estimator. The above sliced covariance estimator is based on the observable  factors. If the factors are unknown, they are replaced by their estimators,  which leads to the following sufficient forecasting algorithm based on the factor models.
\begin{algo}
 Sufficient forecasting algorithm based on the factor models.
\begin{description}
\item[Step 1] Estimate factors  in model (\ref{eqli02})   for $t=1,\dots, T$;
\item[Step 2] Construct the covariance estimator as in (\ref{eqli007}) with $\widehat\bff_t$ in place of $\bff_t$;
\item[Step 3] Obtain $\widehat\bphi_1, \widehat\bphi_2, \dots, \widehat\bphi_R$ by the top $R$   eigenvectors of  the covariance in Step 2;
\item[Step 4] Construct the predictive indices $\widehat\bphi_1'\widehat\bff_t, \dots, \widehat\bphi_R'\widehat\bff_t$;
\item[Step 5] Nonparametrically  estimate $h(\cdot)$ with
indices from Step 4, and forecast $ y_{t+1}$.
\end{description}
\end{algo}

Implementing the above algorithm requires   the number of slices $H$, the number of
predictive indices $R$, and the number of factors $r$. In practice, $H$ has little influence on the
estimated directions, as pointed out in \cite{li1991sliced} and explained above that property (\eqref{eqjf02}) holds.
As regard to the choice of $R$, the first $R$ eigenvalues of $\cov(\mathbb E(\bff_t| y_{t+1}))$ must be significantly different
from zero compared to the estimation error. Several methods such as  \cite{li1991sliced}  and \cite{schott1994} have been proposed to determine $R$. For instance, the average of the smallest $r-L$ eigenvalues would follow $\chi^2$ distribution if the underlying factors are normally distributed. The number of factors can be determined by a number of methods. 

\subsection{Factor-adjusted regularized model selection}

 Consider a high-dimensional regression model
  \begin{eqnarray}\label{eq4.1}
  y_{t} &=& \bbeta'  \bg_t+ \bnu'\bx_t +\eta_t,\cr
  \bg_t&=& \btheta'\bx_t +\bvarepsilon_{g,t}
  \end{eqnarray}
  where  $\bg_t$ is a  treatment variable whose effect $\bbeta$ is of the main interest. The model contains   high-dimensional  exogenous control variables $\bx_t=(x_{1t},\cdots,x_{Nt})$ that determine both the outcome and treatment variables. Having  many control variables creates challenges for statistical inferences, as such,  we assume that  $(\bnu, \btheta)$ are  sparse vectors.   

 Control variables are often strongly correlated due to the presence of  confounding factors
  \begin{equation}\label{eq4.2}
  \bx_t=\bB\bff_t+\bu_t.
  \end{equation}
This  invalidates    conditions of using penalized regressions to directly select among $\bx_t$.  Instead, if we substitute (\ref{eq4.2}) to (\ref{eq4.1}), we reach a  factor-adjusted regression model:
  \begin{eqnarray}\label{eq4.3}
   y_{t} &=&\balpha_y'\bff_t+  \bgamma'\bu_t+   \bvarepsilon_{y,t},\cr
  \bg_t&=& \balpha_g'\bff_t +\btheta'\bu_t+\bvarepsilon_{g,t},\cr
  \bvarepsilon_{y,t}&=&\bbeta'  \bvarepsilon_{g,t}    +\eta_{t}
  \end{eqnarray}
  where  $\balpha_g'=\btheta'\bB$,  $\balpha_y'=\bbeta  \balpha_g'+\bnu'\bB$,  and $\bgamma'=\bbeta  \btheta'+ \bnu'$.  Here $(\balpha_y, \balpha_g,\bbeta)$ are low -dimensional coefficient vectors    while $(\bgamma, \btheta)$ are high-dimensional sparse vectors. Importantly, the model contains  high-dimensional latent controls $\bu_t$, which are weakly dependent due to the nature of idiosyncratic noises.
  The use of $\bu_t$ instead of $\bx_t$ validates conditions for many high-dimensional variable selection methods.  

 \cite{fan2020factor} 
 and \cite{hansen2018fac} showed that the penalized regression can be successfully applied to (\ref{eq4.3}) to  select  components in $\bu_t$, which are    cross-sectionally weakly correlated. Motivated by \cite{belloni2014inference}, the  algorithm can be summarized as follows.  For notational simplicity, we focus on the univariate case $\dim(\bbeta)=1$.

\begin{algo}\label{ag4.1}
Estimate $\bbeta$ as follows.
\begin{description}
\item[Step 1] Estimate $\{(\bff_t, \bu_t): t\leq T\}$  from (\ref{eq4.2})   to obtain $\{(\widehat\bff_t, \widehat\bu_t): t\leq T\}$.

\item[Step 2] Run penalized variable selections on $\widehat\bu_t$:
\begin{eqnarray*}
( \widehat\bgamma,\widehat\balpha_y)&=&\arg\min_{\bgamma,\alpha_y} \frac{1}{T}\sum_{t=1}^T(y_t-\balpha_y'\widehat\bff_t-\bgamma'\widehat\bu_t)^2+ P_{\tau}(\bgamma),\cr
 ( \widehat\btheta,\widehat\balpha_g)&=&\arg\min_{\btheta} \frac{1}{T}\sum_{t=1}^T(\bg_t-\balpha_g'\widehat\bff_t-\btheta'\widehat\bu_t)^2+ P_{\tau}(\btheta).
\end{eqnarray*}

Obtain residuals:
$\widehat \bvarepsilon_{y,t}=y_{t} -( \widehat\balpha_y' \widehat\bff_t+   \widehat\bgamma' \widehat\bu_t),$ and $ \widehat \bvarepsilon_{g,t}=\bg_{t} -( \widehat\balpha_g' \widehat\bff_t+   \widehat\btheta' \widehat\bu_t).$

\item[Step 3] Estimate  $\bbeta$ by residual-regression:
$
\widehat\bbeta=(\sum_{t=1}^T\widehat\bvarepsilon_{g,t} ^2)^{-1}\sum_{t=1}^T\widehat\bvarepsilon_{g,t} \widehat\bvarepsilon_{y,t}.
$
\end{description}
\end{algo}

  Note that $\bgamma:\to P_\tau(\bgamma)$ is a sparse-induced penalty function with a tuning parameter $\tau$. 
  When $\btheta$ and $\bgamma$ are sufficiently sparse,  
  and the PC-estimator is used in  step 1 with the correct selection of the number of factors,
    the above procedure is asymptotically valid:
  \begin{equation}\label{eq3.4}
 \sigma_{\eta, g}^{-1} \sigma_g^{2}  {\sqrt{T}(\widehat\bbeta-\bbeta)}\overset{d}{\longrightarrow}\mathcal N(0,1),
\end{equation}
where $\sigma_g^2$ and $\sigma_{\eta, g}^2$ are the asymptotic variances of $\bvarepsilon_{g,t}$ and $\eta_t\bvarepsilon_{g,t}$.

More recently, \cite{fan2019learning}  showed that the assumption of correct selection of the number of factors can be relaxed if we use the diversified projection in step 1 instead, and  (\ref{eq3.4}) is still valid as long as we select $R\geq r$ factors (over selection).  Importantly, this admits  $r=0$, and $R\geq 1$  as a  special case,  i.e., there are no factors so that $\bx_t=\bu_t$ itself is cross-sectionally weakly dependent, but nevertheless we estimate $R\geq 1$ number of factors to run post-selection inference to alleviate the dependence among $\bx_t$. This setting   is   empirically relevant as it allows  to avoid pre-testing the presence of common factors for inference.

\begin{figure}[h]
\includegraphics[width=3in]{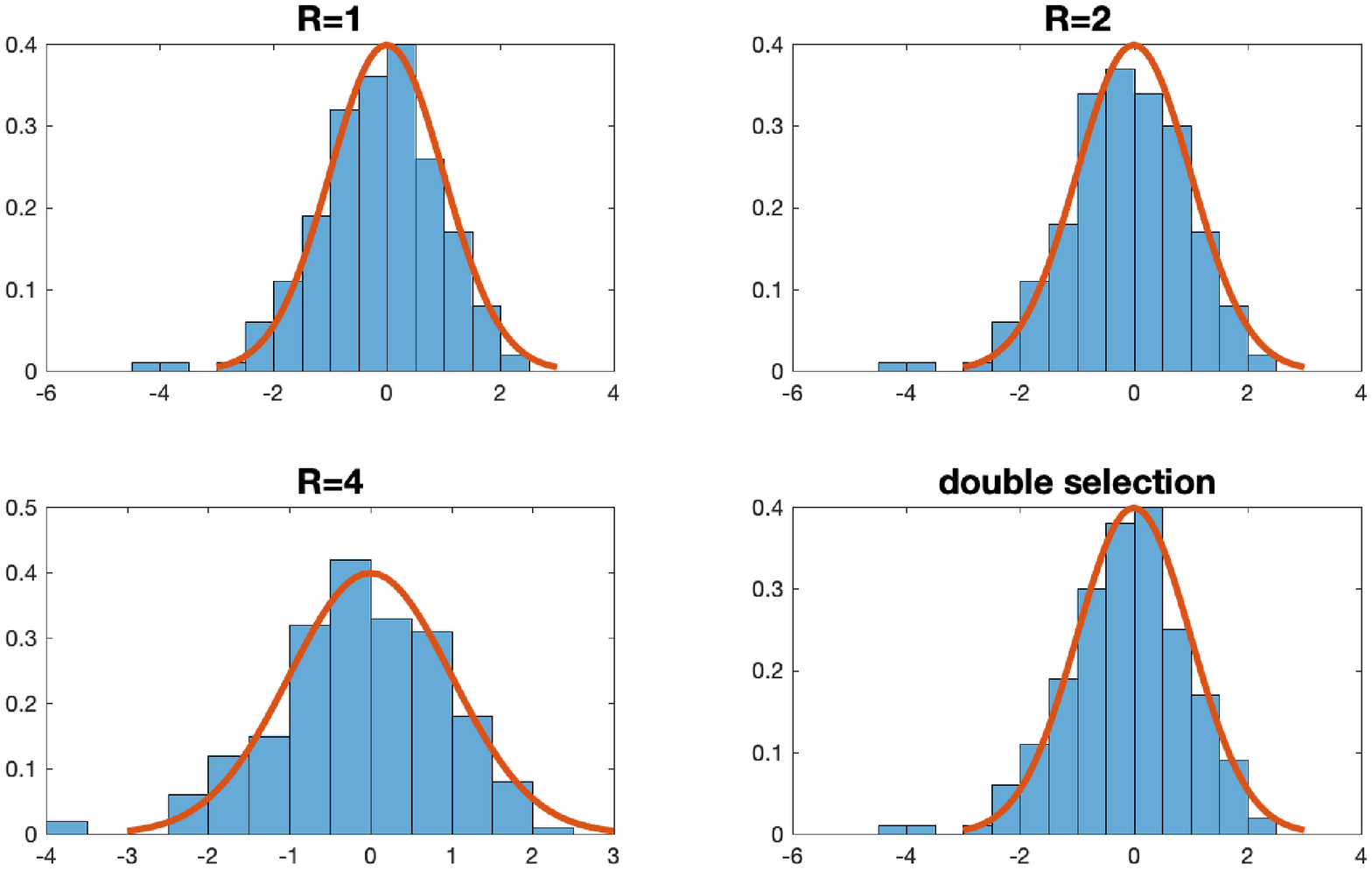}
\includegraphics[width=3in]{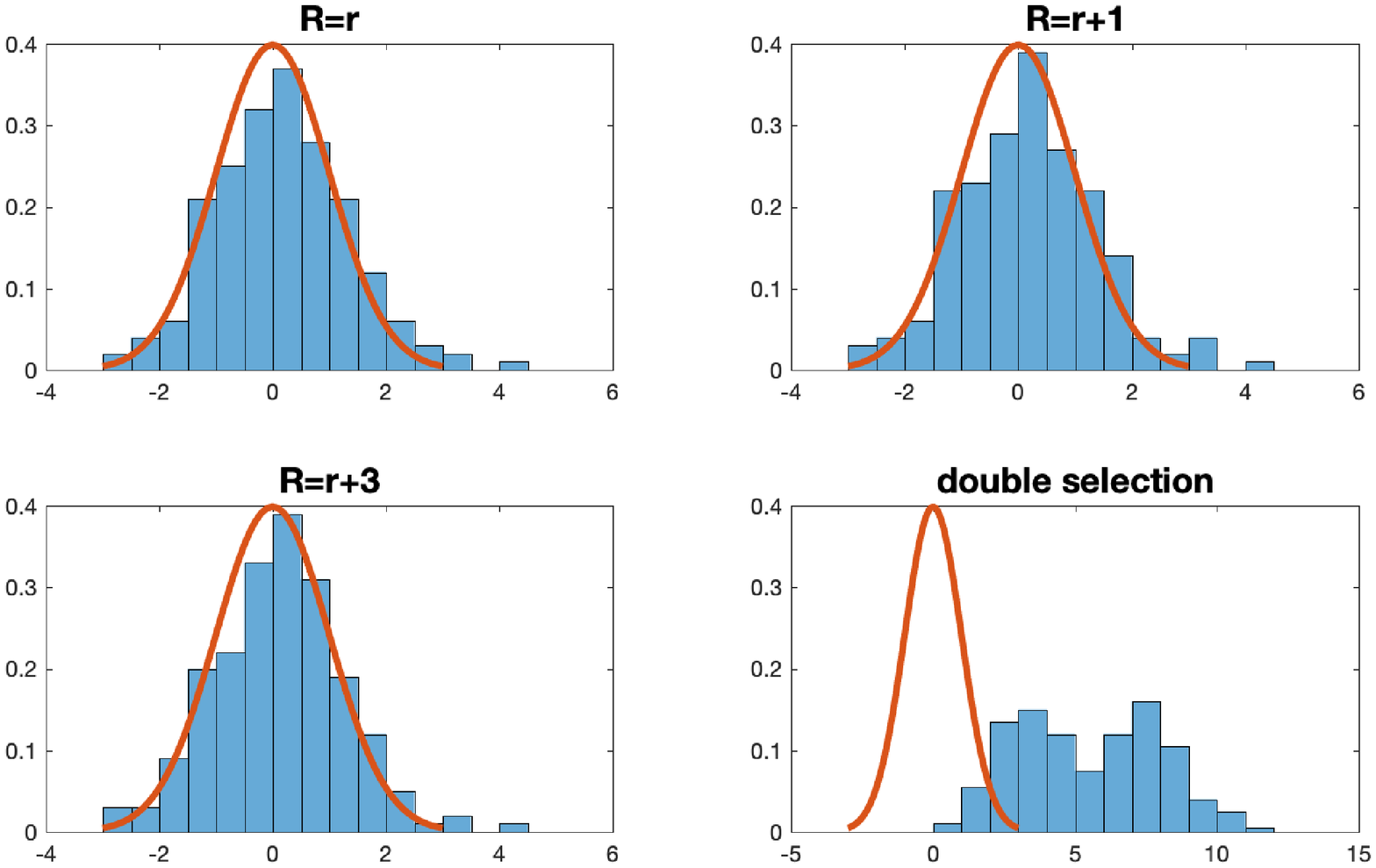}
\caption{Histograms of the standardized estimates (\ref{eq3.4}) over 200 replications, superimposed with the standard normal density. The panel of   ``double selection" corresponds to directly selecting among $\bx_t$ (corresponding to $R=0$, no factor adjustment), while all other panels correspond to using diversified factor (DP) estimators with $R$ number of working factors.   Top four figures correspond to $r=0$ and bottom four figures correspond to $r=2$. When $R \geq r$, (\ref{eq3.4}) holds, whereas when $R < r$, (\ref{eq3.4}) is violated. Figure source: \cite{fan2019learning}.}  
\label{fig1}
\end{figure}

Figure \ref{fig1}, taken from \cite{fan2019learning}, plots the histograms of the t-statistics based on  estimated $\bbeta$ over 200  simulations, superimposed with the standard normal density, where   $R$ diversified projections are used to estimate factors in step 1. Here the weights  are the initial transformations ($t=0$) so   that the $i^{th}$ row of $\bW$ is $ (x_{it}, x_{it}^2,\cdots,x_{it}^R)$ at $t=0$.
  The ``double selection" is the algorithm used in \cite{belloni2014inference} that directly selecting among $\bx_t$, corresponding to the case $R=0$.
The factor-augmented algorithm works well even if $r=0$; but when    $r \geq 1$ factors are present, ``double selection"  leads to severely biased estimations.

Therefore as a practical guidance, we recommend that one  should always run factor-augmented post-selection inference, with  $R\geq 1$, to guard against confounding factors  among the control variables.

  \subsection{Factor-adjusted robust multiple testing}

\subsubsection{False discovery rate control}

Controlling the false discovery proportion  (FDP) in large-scale hypothesis testing based on strongly dependent tests has been an important   problem in many scientific discoveries across disciplines. See \cite{fan2019farmtest} and references therein, and \cite{Barras:2010fh,Harvey2016, harveyliu18, giglio2019thousands} for applications in empirical asset pricing.

Suppose we observe realizations of  a   random vector $\{\by_t=(y_{1t},\cdots,y_{Nt})'\}_{t=1}^T$. Let $\balpha=(\alpha_1,\cdots,\alpha_N)'$ denote its mean vector. We are interested in   testing  individual hypotheses:
$$
H_{0}^i: \alpha_i=0,\quad i=1,\cdots, N.
$$

Let $p_i$ denote the $p$-value for testing   $H_0^i$ based on a test statistic such as $t$-test, which rejects if $p_i<x$ given some critical value $x$.
 Define  the number of false discoveries (rejections)  and   the total number of rejections  as follows:
    \begin{eqnarray*}
  	\mathcal{F}(x) & = & \sum_{i=1}^{N}1\{i: p_{i}<x\text{ and $H_0^i$ is true}\},\qquad
  	\mathcal{V}(x) = \sum_{i=1}^{N}1\{i: p_{i}<x\}.
  \end{eqnarray*}
In large-scale multiple testing problems,        researchers often aim to control the \textit{false discovery proportion} (FDP) and the \textit{false discovery rate} (FDR)   defined by
\[
\text{FDP}(x)= \frac{\mathcal{F}(x)}{\max\{\mathcal{V}(x),1\}},\quad \text{FDR}(x)=\mathbb E \{\text{FDP}(x)\}.
\]
The goal is to find the critical value $x$ so that FDR$(x)\leq\tau$ for a desired level $\tau$ (e.g., 0.10) or more relevantly  FDP$(x)\leq\tau$ with high confidence. While $\mathcal V(x)$ is known, $\mathcal F(x)$ is not in practice. A general principle of finding $x$ proceeds as the following two steps.

\begin{algo}\label{ag45.1}

General principle for FDP/FDR control.
\begin{description}
\item[Step 1.]   Find $\bar{\mathcal F}(x)$ such that    either it upper bounds  $\mathcal F(x)$  for all $x\in(0,1)$, or it estimates $\mathcal F(x)$   uniformly well.

\item[Step 2.]   Set the critical value  to $x^*=\sup\{x\in(0,1): \bar{\mathcal{F}}(x)\leq \tau  \max\{\mathcal{V}(x),1\}\}$.

 \end{description}

 \end{algo}


One of the most popular procedures, proposed by \cite{benjamini1995controlling}, proceeds as follows.    Denote $p_{(1)}\leq \cdots \leq p_{(N)}$ as the sorted p-values  for the individual tests. Then the critical value  is set to
$$x^*=  \max\{p_{(i)}: p_{(i)}\leq \tau i/N \}.
$$
This method  fits into Algorithm \ref{ag45.1} with $\bar {\mathcal F}(x)= Nx $, which is an asymptotic upper bound for $\mathcal F(x)$ when the individual p-values are independent.
One of the limitations of this upper bound is that it is too conservative if the number of true negatives is small compared to $N$. More fundamentally, it requires the test statistics be weakly dependent, a topic we shall discuss in more detail next.
Other methods, such as \cite{storey2002direct, fan2012estimating}, etc., aim to directly estimate $\mathcal F(x)$ in step 1 in the presence of strong dependence among test statistics, and are also adaptive to the unknown number of true negatives.

 In addition, instead of  Algorithm \ref{ag45.1},    \cite{romano2007control,romano2008control} provided  alternative procedures for FDR control.

\subsubsection{Removing dependence by factor adjustments}

The key to the  success of  FDR control   is  that the individual  test statistic  should be either weakly dependent or independent.   This makes the FDR and FDP approximately the same and easier to control. On the other hand, suppose the cross-sectional dependence of   $\by_t $ is  generated from a latent factor model:
\begin{equation}\label{eqjf03}
\by_{t}= \balpha + \bB\bff_t+ \bu_{t}, \quad \E (u_t|\bff_t)=0,
\end{equation}
where $\E \bff_t=0$,  and $\balpha $  is the mean vector.
In empirical asset pricing,  the model can be  used to identify  nonzero alphas out of a large number of assets, and has  been   studied  to identify skilled mutual fund managers, e.g.,  \cite{Barras:2010fh} and \cite{Harvey2016}.  The presence of latent factors, however,  leads to    strong   dependence among the t-statistics based on the naive  sample means of $\by_t$, which  invalidates the weak dependence assumptions.
As well documented in the literature, strong dependence  creates  fundamental challenges to multiple testing, including large standard errors  among the estimated $\alpha_i$, unstable FDP's, and conservativeness of the test procedure.  Learning dependence $\bB \bff_t$ and removing it from the model (\ref{eqjf03}) make the data not only weakly dependent  but also less noisy (from $\bB \bff_t+\bu_t$ to $\bu_t$).  This is the basic idea in factor-adjusted robust multiple tests (FarmTest)  by using factor-adjusted data $\{\by_{t}- \widehat{\bB}\widehat{\bff_t}\}_{t=1}^T$; see (\ref{eqjf03}). Furthermore,  \cite{fan2019farmtest}  makes adjustments  so that it is also robust to heavy tailed data. 

\begin{figure}[h]
\includegraphics[width=3in]{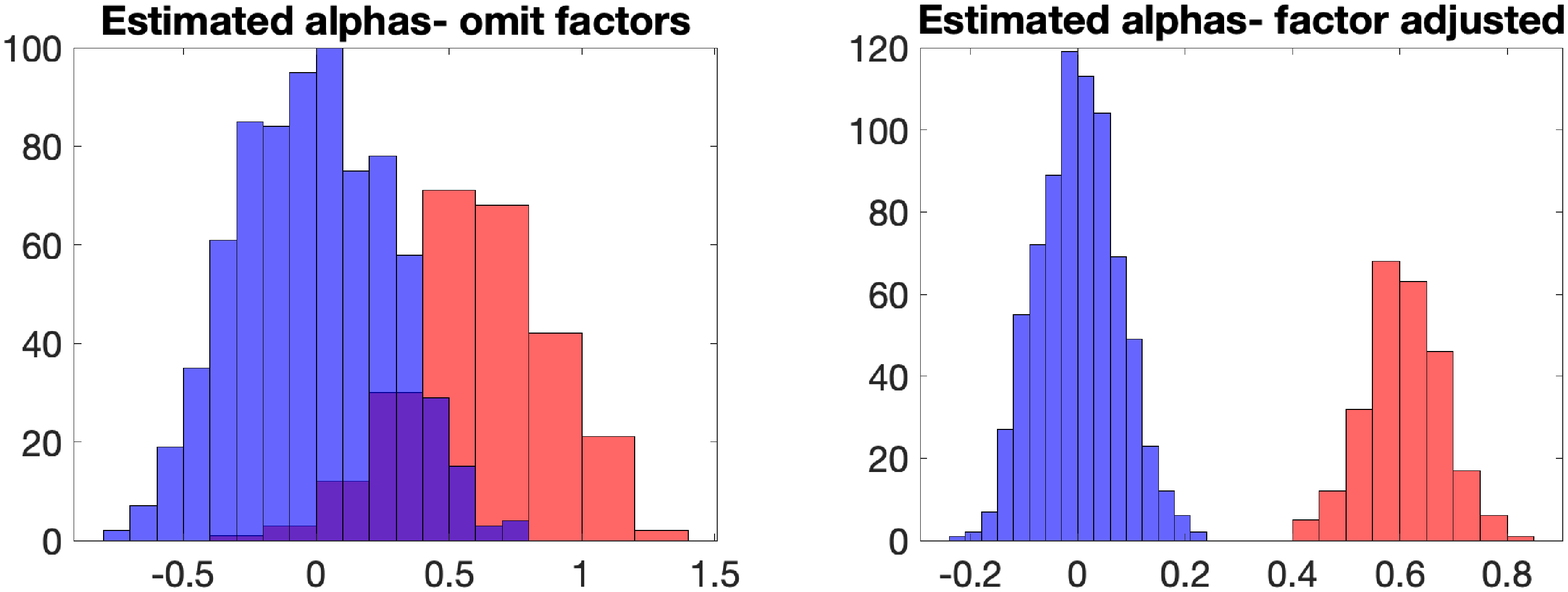}
\includegraphics[width=3in]{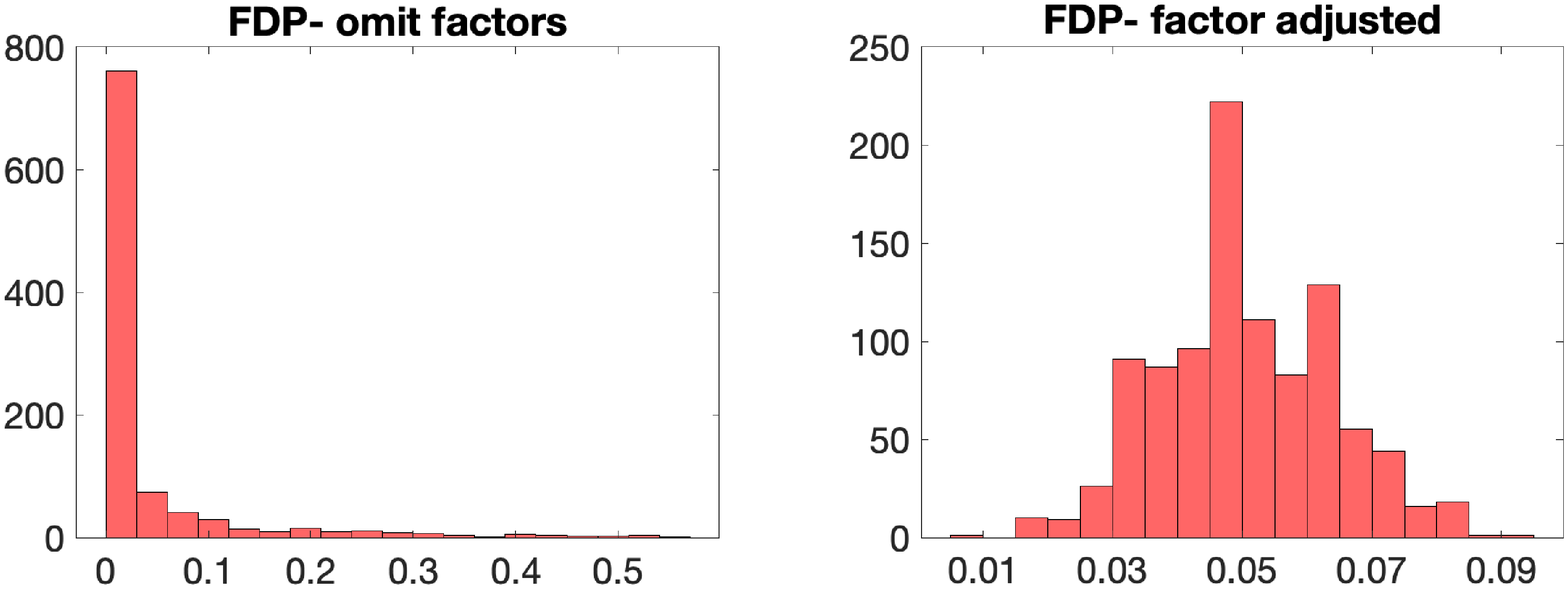}
\includegraphics[width=3in]{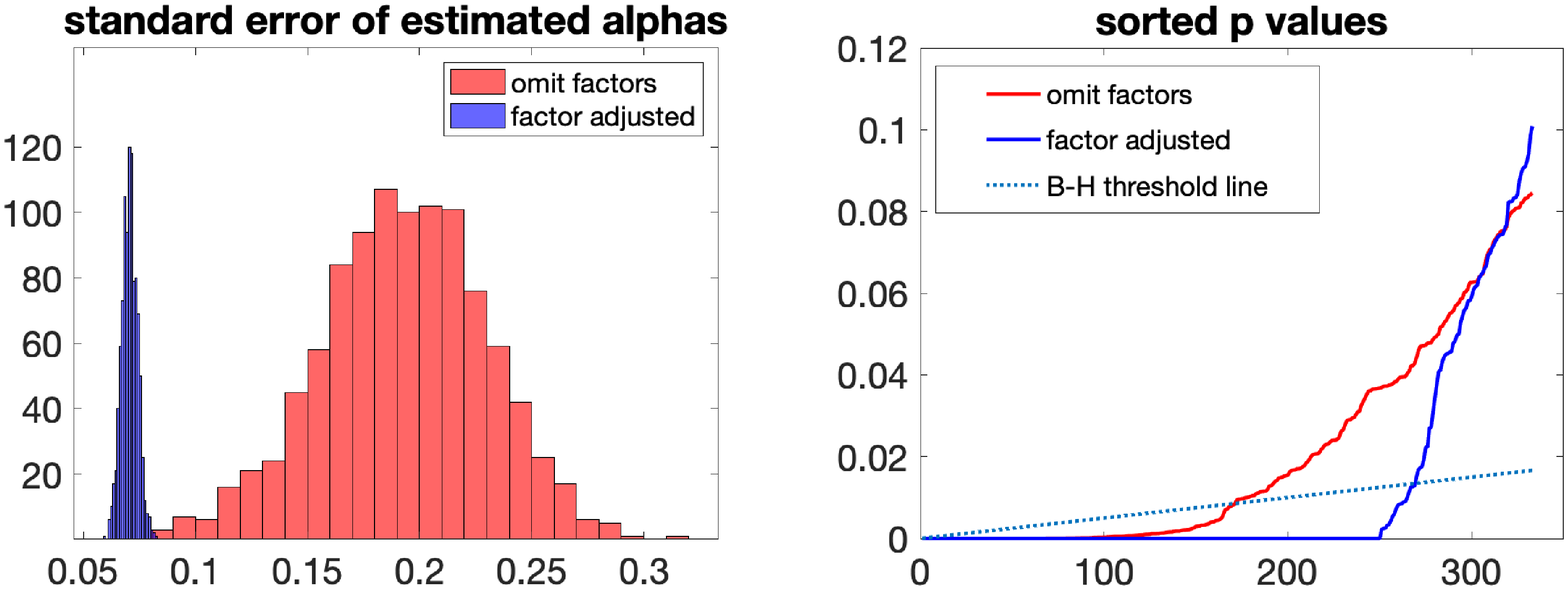}
\caption{Comparison between the sample mean (omit factors) and the factor-adjusted method, with $T=200$ and $N=1000$. The upper panels plot the histograms of estimated individual alphas from a single simulation; the middle panels plot the individual FDP's  over 1000 simulations.  The bottom left panel plots the cross-sectional   histograms of standard errors of the estimated alphas over 1000 simulations. The bottom right panel plots the sorted p-values from a single simulation.  The B-H procedure rejects all the hypotheses if   $p_{(i)}$  is below the B-H threshold line $f(i):=\tau i/N$. }
\label{fig2}
\end{figure}

 To illustrate  consequences of omitting  adjusting latent factors   as well as the   effectiveness of  the use of  the factor-adjusted method (to be detailed below), let us  consider a numerical example of a single factor model,  where elements of $\bu_t$,  $\bff_t$ and $\bB_t$ are generated  from the standard normal distribution. We take the true means to be $\alpha_i=0.6$ for $1\leq i\leq N/4$ and 0 otherwise, and compare two  estimated $\alpha_i$: 1) the sample means  of $\by_t$, without using factor adjustments; 2) the factor-adjusted estimator based on PCA.
We apply the method of \cite{benjamini1995controlling} for multiple testing,  setting $\tau=0.05$.

The top panels of Figure \ref{fig2} plot the histograms, from a single simulation, of the estimators for $\alpha_i$, corresponding to those that satisfy the null hypotheses $\alpha_i=0$ and those that satisfy the alternatives $\alpha_i=0.6$.
Clearly, there is a large overlap (on the upper left panel)
between sample means from the null and alternative,   making tests based on sample means difficult to distinguish
the alternatives from the nulls. 
In contrast, the PCA-based estimator can easily separate the nulls and alternatives, as shown on the upper right panel in Figure \ref{fig2}.

The middle two panels of Figure \ref{fig2} plot the histograms of the true FDP over 1000 simulations based on  the two estimators.
 It is evident that the distribution of the FDP corresponding to the factor-adjusted estimator concentrates around the nominal level.  In contrast, the one based on the sample mean has a noticeable long tail as well as a larger mean and variance, which demonstrate the challenge to control FPD in presence of common factors, as explained above.

 Finally, omitting confounding factors would lead to larger standard errors and conservative inference.  The bottom two panels in Figure  \ref{fig2} plot the standard errors of individual estimated alphas and the sorted p values for the two estimation methods. The sample-mean estimator  has much   fewer sorted p-values below the B-H threshold line (i.e., fewer rejections), compared to the factor-adjusted estimator.

 Hence it is recommended to estimate and remove the   latent factors before applying standard FDR control algorithms. 

\subsubsection{Identifying skilled hedge funds}

\cite{giglio2019thousands} studied the problem of identifying hedge  funds that are able to produce positive alphas (i.e., have ``skill"), among thousands of existing funds. They considered a  linear   pricing model, where hedge fund returns are:
$$
y_{it}=\alpha_i+\bb_i'\blambda+\bb_i'(\bff_t-\mathbb E\bff_t)+u_{it}.
$$
In the model $\bff_t$ contains both observable and latent factors. The model allows  nontradable observable factors and $\blambda$ is the vector of  factor risk premia.

At a broad level, their  methodology proceeds as   the Fama-MacBeth regression integrated with the PCA to extract latent factors:

\begin{algo}\label{ag45.1adf}

Estimating alphas in the presence of latent and nontradable factors.

\begin{description}
\item[Step 1.] Run fund-by-fund time series
regressions to estimate fund exposures (betas) to  observable   factors.

\item[Step 2.] Apply PCA to the residuals to recover the latent factors and   betas.

\item[Step 3.] Implement cross-sectional regressions like Fama-MacBeth to estimate the risk premia of
the factors (including   both observable
and  latent factors) and the alphas.

\end{description}

 \end{algo}

Because of many negative alphas from unskilled hund managers, the multiple testing problem should be properly formulated as one-sided hypotheses:
$$
H_{0}^i: \alpha_i\leq 0,\quad i=1,\cdots, N.
$$
Hence  rejecting $H_0^i$ indicates skilled fund manger $i$. On the other hand,    the existence of  potentially a very large number of negative alphas gives rise to the issue of power loss, only to add   noises to the model. The loss of power associated with testing inequalities is well known as  the problem of ``deep in the null", and is often seen in the econometric literature.   To address this issue,
\cite{giglio2019thousands} proposed to first screen off very bad funds, identified as:
$$
\mathcal I=\{i\leq N,\widehat\alpha_i/\text{se}(\widehat\alpha_i)<-c_{NT}\}
$$
where $c_{NT}>0$ is a slowly growing sequence to ensure sure screening \citep{fan2008sure}: $P(\mathcal I\subseteq \mathcal H_0)\to 1$. They recommended to apply   FDR control algorithms on funds outside $\mathcal I$. Therefore,
there are two ingredients that are recommended for   identifying skilled fund managers via multiple testing:
(1) adjust the effect of latent factors, and (2) remove the estimated alphas that are deep in the null.  Both are  playing essential roles of  gaining good testing power.

\subsection{Instrumental variable regression} The  issue of endogeneity is often encountered in real data applications. 
Consider  the following instrumental variable (IV) regression model
\[y_t=\bw_t'\bbeta^0+\varepsilon_t=\bw_{1t}'\bbeta_1+\bw_{2t}'\bbeta_2+\varepsilon_t,\]
where $\bw_{1t}$ is a $k_1$-dimensional vector of exogenous regressors and $\bw_{2t}$ is a $k_2$-dimensional vector of endogenous regressors. 
 Meanwhile, we have an $N$-dimensional IV $\bx_t$ which admit a factor structure: 
$$
\bx_t=\bB\bff_t+\bu_t.
$$

 Below we introduce four  estimators for $\bbeta^0$, which differ on their choices of the instruments.
 
 \textbf{Use $\bff_t$  as the instruments.}  Project $\bw_{2t}$ on $\bff_t$:
\[\bw_{2t}=\bphi'\bff_t+\bv_{t},\quad \mathbb E(\bv_t|\bff_t)=0\]
where $\bphi$ is a $k_2\times r$ matrix. 
We need $r\ge k_2$ for identification.  Let $\bz_t=(\bw_{1t}', \bff_t')'$ be the set of instruments. As $\bff_t$ is unobservable, we replace   it with some factor estimator and apply the two  stage least squares estimator $\widehat\bbeta_\bff$ with the feasible instruments.

\cite{bai2010}  studied this estimator, and showed that   the estimation errors of  $\widehat\bbeta_\bff$ associated with the generated instruments (factor estimations) have no effect on the limiting variance.   When $\bu_t$ and $\varepsilon_t$ are uncorrelated,  this  only requires  $(N,T)\to\infty$ regardless of the relative growth rates.  When some weak correlations are present but   $\|\mathbb E \varepsilon_t \bu_{t}\|_1=O(1)$,  we would require $\sqrt{T}=o(N)$ to offset the effect of estimating factors.



 \textbf{Use $\bx_t$  as the instruments.}   Project $\bw_{2t}$ on $\bx_t$:
\begin{equation}\label{eq4.11}
\bw_{2t}=\btheta\bx_t+\be_{t},
\end{equation}
 where $\btheta$ is a  $k_2\times N$ coefficient matrix.  This projection motivates the use of $\bx_t$ directly as   a set of high-dimensional IV.   Suppose that $\varepsilon_t$ is an i.i.d process, then the two-stage least squares estimator is efficient, and is given by
\[\widehat\bbeta_{\bx}=\Big(\bW'\bX\widehat  \bSigma_x^{-1}\bX'\bW\Big)^{-1}\bW'\bX\widehat \bSigma_x^{-1}\bX'\bY\]
where $\bX$ is $T\times N$ matrix of $\bx_t$; $\bW$ and $\bY$ are matrices of $\bw_t$ and $y_t$.
Note that $\widehat \bSigma_x$ is the  estimated covariance of $\bx_t$, which can be constructed using factor-based covariance estimators as described in Section 3. 
  It is interesting to compare the asymptotic behaviors of   $\widehat\bbeta_\bff$  with $\widehat\bbeta_{\bx}$.
\cite{bai2010} showed when $\bu_t$ and $\bw_t$ are uncorrelated,  they  have the same asymptotic variance, but  $\widehat\bbeta_{\bx}$  has a $O(\frac NT)$ bias term. So  $\widehat\bbeta_{\bx}$ is consistent only if $N=o(T)$. 


 \textbf{Use selected $\bx_t$  as the instruments.}   We still consider the projection (\ref{eq4.11}), but assume that rows of $\btheta$ are sparse vectors so that we can apply penalized regression to select among the components of $\bx_t$:
 \begin{equation}\label{eq4.12}
 \widehat\btheta_j=\arg\min_{\btheta_j\in\mathbb R^N} \frac{1}{T}\sum_{t=1}^T(w_{2t,j}- \bx_t'\btheta_j)^2+P_\tau(\btheta_j),\quad j\leq \dim(\bw_{2t})
 \end{equation}
 where $P_\tau(\btheta_j)$ is a sparse-induced penalty with tuning $\tau$.  Let $ \bx_{t, \text{selec}}$ be  the vector of  selected components corresponding to nonzero components of $\{\widehat\btheta_j: j\leq \dim(\bw_{2t})\}$. \cite{belloni2012sparse}   used $(\bw_{1t}, \bx_{t, \text{selec}})$ as the instruments to compute $\widehat\bbeta_{\bx,\text{selec}}$, the two stage least squares estimator.  This method however, would  not work well in the presence of common factors. The strong dependence  in   $\bx_t$   invalidates  the    variable selection procedure (\ref{eq4.12}).

  \textbf{Use $\bff_t$ and selected $\bu_t$  as the instruments.} We are not aware of any applications of this method in the IV literature, but it is still well motivated.  Substitute the factor structure to (\ref{eq4.11}), we obtain 
  $$
  \bw_{2t}= \bdelta\bff_t+\btheta\bu_t+\be_t
  $$
  where $\bdelta=\btheta\bB$. Hence we can carry out variable selections  among $\bu_t$: 
   $$
 (\widehat\bdelta_j, \widehat\btheta_j)=\arg\min_{\bdelta_j\in\mathbb R^r, \btheta\in\mathbb R^N} \frac{1}{T}\sum_{t=1}^T(w_{2t,j}- \bdelta_j'\widehat\bff_t-\widehat\bu_t'\btheta_j)^2+P_\tau(\btheta_j),\quad j\leq \dim(\bw_{2t}).
 $$
 Let $ \widehat\bu_{t, \text{selec}}$ be  the vector of  selected components corresponding to nonzero components of $\{\widehat\btheta_j: j\leq \dim(\bw_{2t})\}$. We then use  $(\bw_{1t}, \widehat\bff_t, \widehat \bu_{t, \text{selec}})$ as the instruments to compute $\widehat\bbeta_{\bff, \bu}$, the two stage least squares estimator.  This method is expected  to work well because it marginalizes out the strong factors in $\bx_t$, leaving remaining components $\bu_t$ being weakly dependent.

Let us conduct a simple simulation to study the finite sample behaviors of the aforementioned four estimators.  We consider a model  $y_t=\bw_{2t}'\bbeta^0+\varepsilon_t$, with a single endogenous regressor  $\bw_{2t}$ generated  from (\ref{eq4.11}) with $\be_t= \varepsilon_t/2$. Here $\btheta=(2,1,-1,0...,0)$ and $\bx_t$ admits a two-factor structure. Variables $(\varepsilon_t, \bff_t, \bB,\bu_t)$ are independent  standard normal. Finally, variable selections are based on lasso with the oracle tuning parameter that controls the score of the least squares function. For instance, for problem (\ref{eq4.12}) we set $P_\tau(\btheta)=\tau\|\btheta\|_1$ with $\tau= 2.2\|\frac{1}{T}\sum_t\bx_t\be_t\|_\infty$. 

\begin{table}[h]
\tabcolsep7.5pt
\caption{Comparison among four IV   methods}
\label{tab2}
\begin{center}
\begin{tabular}{cc|cccc|cccc}
\hline
\hline
$N$&$T$ & \multicolumn{4}{c|}{Bias}&\multicolumn{4}{c}{Standard deviation}  \\
 &  & $\widehat\bbeta_{\bff}$  &$\widehat\bbeta_{\bx}$&  $\widehat\bbeta_{\bx,\text{selec}}$&$ \widehat\bbeta_{\bff, \bu} $&$\widehat\bbeta_{\bff}$&$\widehat\bbeta_{\bx}$  &$\widehat\bbeta_{\bx, \text{selec}}$&$ \widehat\bbeta_{\bff, \bu}$ \\
\hline
&&&&&&&&&\\
50&100& 0.004  & -0.993& 0.003& 0.008& 0.215& 0.001&  0.061&0.058\\
200&100& 0.008&  -0.996&  0.007& 0.009& 0.143& 8e-4& 0.056& 0.054\\
\hline
\end{tabular}
\end{center}
\begin{tabnote}
Reported is based on  1000 replications. $\widehat\bbeta_\bff$ uses $\widehat\bff_t$ as IV;  $\widehat\bbeta_\bx$ uses $\bx_t$ as IV; $\widehat\bbeta_{\bx, \text{selec}}$ selects $\bx_t$ as IV  using lasso; $\widehat\bbeta_{\bff, \bu}$ uses $(\widehat\bff_t, \widehat\bu_{t,\text{selec}})$ as IV, where $\widehat\bu_t$  are selected using lasso.\end{tabnote}
\end{table}

 Table \ref{tab2} reports the bias and standard deviation of each    estimator calculated from  1000 replications.   First,  using only estimated factors as the instruments ($\widehat\bbeta_\bff$) leads to the largest standard error. This is not surprising because 
 it excludes the relevant information from $\bu_t$ while the latter is correlated with $\bw_{2t}$, so this method  is less efficient. 
 Secondly,  using $\bx_t$  as instruments without variable selection ($\widehat\bbeta_\bx$) has the smallest standard deviation, but is   severely biased.  
 Finally, the two instrumental selection based estimators ($\widehat\bbeta_{\bx,\text{selec}}$ and $\widehat\bbeta_{\bff,\bu}$) perform favorably and similarly. But $\widehat\bbeta_{\bx,\text{selec}}$ is not as stable, as it occasionally selects none of the  instruments in our numerical experiments. 
 

\subsection{Boosting}


Consider the following factor-augmented regression
\begin{equation}\label{eqli07} y_{t+h}= c+\balpha(\bL)'\bw_t+\bgamma(\bL) y_t+\bbeta(\bL)'\bff_t+\varepsilon_{t+h}.\end{equation}
where $\balpha(\bL)=\balpha_0+\balpha_1\bL+\dots+\balpha_p\bL^p$, $\bgamma(\bL)=\bgamma_0+\bgamma_1\bL+\dots+\bgamma_q\bL^q$ and $\bbeta(\bL)=\bbeta_0+\bbeta_1\bL+\dots+\bbeta_l\bL^l$, all are lag operator polynomials. Suppose that $\bw_t$ is a $k$-dimensional vector and $\bff_t$ is an $r$-dimensional vector. The above predictive regression has $n=1+(p+1)k+(q+1)+(l+1)r$ parameters. It is likely that partial parameters are zero. So   model selection devices   can be conducted to choose a parsimonious model. Here we briefly describe a model selection method, known as \textit{boosting}, which was proposed to use by \cite{bai2009boosting} in this context.

Boosting is an ensemble meta-algorithm, which sequentially finds a ``committee'' of base learners and then makes a collective decisions by using a weighted linear combination of all base learners. The first successful and popular boosting algorithm is \textit{AdaBoost} \citep{freund1997}.
\cite{friedman2001} proposes a generic functional gradient descent (FGD) algorithm, which views the boosting as a method for function estimation.
If the squared loss     function is specified, the FGD algorithm reduces to the $L_2$-Boosting, which is studied in \cite{friedman2001} and \cite{buhlmann2003}. Suppose that $( y_t, \bz_t)_{t=1}^T$ are the observed target and predictive regressors over the sample period. The $L_2$-Boosting algorithm for estimating the conditional mean $\mathbb E (y_t|\bz_t)$ is given as follows.

\begin{algo}\label{algo4.2}
 $L_2$-Boosting algorithm
\begin{description}
\item[Step 1] Initialize $\widehat f^{[0]}(\cdot)$ an offset value. The default value is $\widehat f^{[0]}(\cdot)\equiv \bar y$. Set $m=0$.
\item[Step 2] Increase $m$ by 1. Compute the residuals $e_t= y_t- \widehat f^{[m-1]}(\bz_t)$ for $t=1,2,\dots, T$.
\item[Step 3] Fit the residual vector $e_1, \dots, e_T$ to $\bz_1, \dots, \bz_T$ by the real-valued base procedure (e.g., regression):
\[(\bz_t, e_t)_{t=1}^T\xlongrightarrow{base~procedure} \widehat g^{[m]}(\cdot).\]
\item[Step 4] Update $\widehat f^{[m]}(\cdot)=\widehat f^{[m-1]}(\cdot)+\nu\cdot \widehat g^{[m]}(\cdot)$, where $0<\nu\le 1$ is a step-length factor.
\item[Step 5] Iterate steps 2 to 4 until $m=m_{\mathrm{stop}}$ for some stopping iteration $m_{\mathrm{stop}}$.
\end{description}
\end{algo}

One can apply the above $L_2$-Boosting to the factor-augmented predictive regression (\ref{eqli07}). As seen in Algorithm \ref{algo4.2}, one needs to specify the base procedure in step 3. \cite{bai2009boosting} suggest two methods depending on the way to deal with lags, which leads to the component-wise $L_2$-Boosting and block-wise $L_2$-Boosting. In component-wise $L_2$-Boosting, one treats each lag of each variable as an independent predictor and the base procedure is a simple linear regression. Therefore, step 3 is given as follows.

\begin{algo}
Component-wise $L_2$-Boosting
\begin{description}
\item[Step 3.1]  Let $\bz_{t,j}$ denote a typical regressor in the regressors pool with $j=1,2,\dots, n$. Regress the current residual $e_t$ (the residual in the $m$-th repetition) on each $\bz_{t,j}$ to obtain the coefficient $\widehat \bb_j$. Compute the sum of squared residuals, denoted by SSR($j$).
\item[Step 3.2] Determine $j_m$ by
\[j_m=\underset{1\le j\le n}{\mathrm{argmax}}~ \mathrm{SSR}(j).\]
\item[Step 3.3] $\widehat g^{[m]}(\bx_t)=\bz_{t,j_m}\widehat \bb_{j_m}$ if $\bx_t=\bz_{t,j_m}$, and 0 otherwise.
\end{description}
\end{algo}

Another way is to only differentiate the predictors in the current period and treat the predictor and its multiple lags as a block. This gives rise to the block-wise $L_2$-Boosting. 
The base procedure now is a multivariate regression with the regressors being one predictor and its lags.  See \cite{bai2009boosting} for details. 

\subsection{Threshold regression with mixed integer optimization}

 Threshold  regressions have been used in economic applications to capture potential   structural changes on regression coefficients. The early literature models the threshold effect using some observable scalar variable $q_t$ as in:
  $$
y_{t}= \bw_t'\bbeta+\bw_t'\bdelta1\{q_t>\gamma\}+\varepsilon_t,
 $$
 where $\bw_{t}$ and $q_{t}$ are adapted to the filtration $\mathcal{F}_{t-1}$;
 $(\bbeta, \bdelta, \gamma)$ is a vector of unknown parameters,
 and  $\varepsilon _{t}$ satisfies the conditional mean restriction.
Hence when  $q_t > \gamma$, the regression function becomes
 $\bw_{t}^{\prime }(\bbeta +\bdelta)$;  when $q_t \leq \gamma$, it reduces to $\bw_{t}^{\prime }\bbeta$ \citep{chan1993consistency, hansen2000sample}.
  In practice,
 it might be controversial to choose which observed variable plays the role of $q_t$. For example, if the two different regimes represent
 the status of two   environments of the population,  arguably it is difficult to assume that the change of the environment is governed by just a single variable.

\cite{Seo-Linton} and  \cite{leemyung} extended the model to  multivariate threshold:
 $$
y_{t}= \bw_t'\bbeta+\bw_t'\bdelta1\{\bgamma'\bff_t>0\}+\varepsilon_t,
 $$
 where $\bff_t$ is a vector of ``factors" and $\bgamma$ is the corresponding unknown coefficients. So the model  introduces a regime change due to a single index of factors. Allowing multivariate thresholding is important, because it permits the structural change to be governed by  a potentially much larger dataset:
 $
 \bx_t=\bB\bff_t+\bu_t,
 $  where $\dim(\bx_t)=N\to\infty$.
 So $\bff_t$ can be unobserved factors that can be learned from $\bx_t.$
  For the identification purpose, suppose $\frac{1}{T}\sum_t\bff_t\bff_t'=\bI$ and $\bB'\bB$ is diagonal, then $\bgamma$ and $\bff_t$ are separately identified.  This gives rise to the \emph{factor-driven two-regime regression model}.

  A natural strategy to estimate  the model is to rely on least squares:
   $$
   \min_{\bbeta,\bdelta,\bgamma}\sum_{t=1}^T(y_t-\bw_t'\bbeta-\bw_t'\bdelta1\{\bgamma'\widehat\bff_t>0\})^2,
   $$
   where $\widehat\bff_t$ is the plugged-in PC-estimator of factors. Because
the least squares problem is neither convex nor smooth in $\bgamma$, the computational task is demanding.  
 \cite{leemyung} recommended    using algorithms  based on  mixed integer optimization (MIO).   Introduce integers
  $d_t:=1\{\bgamma'\widehat\bff_t>0\}\in\{0,1\} $.
  The goal is to introduce  linear constraints with respect to variables of optimization.
  Suppose
  there are   known upper and lower bounds for $\delta_{j}$:  $L_j \leq \delta_{j} \leq  U_j$,
  where $\delta_{j}$ denotes the $j$th element of $\bdelta$.
 Define $  M_t \equiv \max_{\bgamma\in\Gamma} |  \bgamma'\widehat\bff_t|  $, where $\Gamma$ is the parameter space for $\bgamma$.  Then it can be verified that the least squares problem is numerically equivalent to the following constraint MIO problem:   \begin{align}\label{prob3}
 \min_{\bbeta, \bdelta, \bgamma, \bd, \bell }
  \sum_{t=1}^{T}
 (y_{t}-\bw_{t}'{\bbeta}- \bw_t'\bell_t )^{2}
 \end{align}
 subject to (for any $\epsilon>0$),  for each $t=1,\ldots,T$ and each $j=1,\ldots,\dim(\bw_t)$,
 \begin{align}\label{main-constraints}
 \begin{split}
 &  \bgamma \in \Gamma,  \; d_t \in \{0, 1\}, \; L_j \leq \delta_{j} \leq  U_j, \\
 & (d_t - 1) (M_t + \epsilon) <   \bgamma'\widehat\bff_t \leq d_t M_t, \\
 & d_t L_j \leq \ell_{j,t} \leq d_t U_j, \\
 & L_{j} (1-d_t) \leq \delta_{j} - \ell_{j,t} \leq U_{j} (1-d_t).
 \end{split}
 \end{align}
Then, we can apply  modern MIO packages (e.g., Gurobi) to solve  for the optimal $(\bbeta,\bdelta,\bgamma)$.

Finally,  \cite{leemyung} also derived the asymptotic distribution of the estimated coefficients and proposed inferences based on bootstraps.  Under the condition that $T=O(N)$, they showed that the effect estimating factors  is negligible on the asymptotic distribution of the estimated $(\bbeta,\bdelta)$, but would affect both the rate of convergence and the limiting distribution of the  estimated $\bgamma$.

\subsection{Community detection}\label{sec:commu}

The stochastic block model has been a  popular approach to modeling networks (see \cite{abbe2017community} for a recent review). We observe a graph of $N$ nodes. Let  $\bA=(a_{ij})\in\mathbb R^{N\times N}$ be the adjancy matrix of edges so that $a_{ij}=1$ if nodes $i$ and $j$ are connected, and $a_{ij}=0$ otherwise.   Suppose each node  belongs to one of $r$ communities,  and the community that node $i$ belongs to is denoted by an unknown $\pi_i\in\{1,\cdots,r\}$.  In addition, elements of $\bA$ are random variables. Then stochastic block model assumes that
$$
P(a_{ij}=1|\pi_i=k, \pi_j = l)= w_{k,l},
$$
where $w_{k,l}$ is an unknown probability. We observe the matrix $\bA$ and aim to recover the membership $\pi_i$ and the probabilities $w_{k, l}$ for all $k, l= 1, \cdots, r$.

Let $\be_1,\cdots,\be_r$ denote the canonical basis in $\mathbb R^r$, and $\bb_i=\be_k$ where $\theta_i=k$.
Then, $\bb_i$ indicates the community membership of node $i$, and the membership matrix is
$$
  \bB=(\bb_1,\cdots,\bb_N)',\quad  N\times r,
$$
whose rows represent nodes and columns represent communities.
Let $\bW$ denote the  $r\times r$ matrix of $(w_{k,l})$ and let   $\bL:=\E\bA$. It can easily be seen that       $\bL=\bB \bW \bB' $ is a low-rank matrix, whose rank equals $r$, leading to the following low-rank decomposition:
$$
\bA= \bL+ \bS,\quad \bS= \bA-\E \bA.
$$
Therefore, $\bA$ has the familiar decomposition (\ref{eq2.2}), with $\bL$ being similar to the systematic risk and $\bB$ as a low-rank loading matrix.  Since the elements in $\bS$ are independent with mean-zero (Wigner matrix), the operator norm $\|\bS\|$ does not grow too fast, compared to that of $\bL$.
We can then apply PCA on $\bA$ to estimate $\bB$. Suppose $r$ is known, then the estimator $\widehat\bB$ is defined as $\sqrt{N}$ times the eigenvectors of $\bA$, corresponding to the first $r$ eigenvalues.

Theorem \ref{th1.1} can be applied to obtain a deviation bound for the estimated loading matrix.
If there  is a sequence $g_{N}\to\infty$ and constants $c_1,\cdots,c_r>0$ such that  the eigenvalues $\lambda_i (\bW^{1/2}\bB'\bB\bW^{1/2}) = c_i g_N (1+o_P(1))$ for all $i\leq r$, then there  is an $r\times r$ matrix $\bH$, so that
$$
 \|\widehat\bB-\bB\bH\|_\infty
=O_P(   g_N^{-2} N \|\bS\|
+ g_N^{-1} \sqrt{N\log N}).
$$



Therefore, elements of a rotated $\bB$ can be estimated uniformly well.
  Moreover, because each community has many nodes belong to, $\bB\bH $ has many identical rows, which makes the cluster analysis as a natural method for community detections. For instance,  we can apply either  the K-means cluster analysis, or the \textit{homogeneous pursuit} of \cite{ke2015homogeneity} on the rows of $\widehat\bB$ to consistently identify the communities.



\subsection{Time varying models}

So far we have been assuming  that the  factor loading  and covariance matrices are time-invariant. Research on  conditional factor models has also grown  rapidly in recent years. Suppose
$$
y_{it} = \bb_{i,t}'\bff_t+ u_{it}
$$
where $\bb_{i,t}$ is a time-varying vector of  loadings.   There have been several approaches to addressing  the issues   of  time-varying loadings. In this section we briefly review three of the most commonly used ones: (1) time-varying characteristics, (2) time-smoothing and (3)  continuous-time models.

\subsubsection{Time-varying characteristics}

The first approach models $\bb_{i,t}$ using a function of observed characteristics $\bz_{i,t-1}$:
$$
\bb_{i,t}= \bb_i(\bz_{i,t-1})
$$
where $\bb_i(\cdot)$ is either a linear function or an unknown   nonparametric function of the characteristics. Therefore, the time-varyingness is mainly captured by the characteristics. An advantage of this approach, over the other two approaches to be reviewed below, is that if $\bz_{i,t-1}$ is correctly specified and indeed can fully capture the degree of time-varyingness of the model,    then $\bb_{i,t}$ allows  a large degree of varyingness, and potentially, structural breaks. On the other hand, the limitation of this approach is the potential misspecification of $\bz_{i,t-1}$ and omitted variable problems.
Above all, we refer to   \cite{gagliardini2019estimation} for an excellent   review on conditional factor models using this approach, and their applications in empirical asset pricing.

\subsubsection{Time-smoothing}

 The second approach assumes that factor loadings change smoothly over time. Suppose $\bb_i(\cdot)$ is an unknown smooth function, we assume
 $$
 \bb_{i,t}= \bb_i\left(\frac{t}{T}\right),\quad \forall t\leq T.
 $$
Then locally, $  \bb_{i,t}   \approx  \bb_{i,r}   $  for all $t\approx r$. So in a local window $\mathcal B(r)$ of each fixed $r$, the model is approximately time invariant:
$$
y_{it}\approx \bb_{i,r}'\bff_t+ u_{it},\quad t\in\mathcal B(r).
$$
Motivated by this assumption, \cite{ang2012testing} and \cite{ma2020testing} tested the market mean-variance efficiency assumption in the case of known factor case. In the unknown factor case, \cite{su2017time}  first applied local smoothing on $y_{i,t}$ then employed PCA on the smoothed data to estimate the factors and loadings.  While this approach does not require the specification of time-varying characteristics, it restricts to the smooth varying scenario and thus rules out structural breaks.   In addition, slow rates of convergence appear  near boundaries (that is, the beginning and the end of observing periods).

\subsubsection{High-frequency factor models}

Consider a continuous-time factor model
$$
d\by_t=\balpha_tdt+\bB_td\bff_t+d\bu_t
$$
where $\by_t$, $\bff_t,\bu_t$ are vectors of asset prices, factors and idiosyncratic risks; $\balpha_t$ is a drift term. The time-varying loading matrix $\bB_t$ is an $N\times r$ matrix that is assumed to be continuous and locally bounded It\^{o} semimartingale of the form:
$$
\bB_{t}=\bB_{0} +\int_0^t\widetilde\balpha_sds+\int_0^t\bsigma_{s}d\bW_s,
$$
where $\widetilde\balpha_s$ and $\bsigma_s$ are optional processes and locally bounded; $\bW_s$ is a Brownian motion. Roughly speaking, by the Burkholder-Davis-Grundy inequality (cf. chapter 2 of \cite{jacod2011discretization}), $\bB_t$ is also locally time-invariant, which is similar to the treatment of the time-smoothing approach. The major difference though, is that the use of high-frequency data has automatically ``smoothed" the data.  We refer to the following papers for recent developments on high-frequency factor models, among others:
  \cite{ait2017using, chen2019five, liao2018uniform,li2019jump, pelger2019large}.

\section{Unbalanced Panels}

Missing data and unbalanced panels  are not uncommon in economic and financial studies. Addressing the missing data issue in statistical modeling belongs to a larger category of problems, known as \textit{matrix completion}.
Low-rank matrix completion refers to the problem of  recovering missing entries from low-rank matrices.
It is    particularly relevant to empirical asset pricing   factor models, because many time series of returns have short histories or missing records. In this section we review several  methods for matrix completions, which assume that the missing is at random, except for  \cite{cai2016structured,bai2019matrix}. Besides, the EM algorithm is also a classical approach to dealing with unbalanced panels. We refer to \cite{stock2002macroeconomic, su2019factor,  zhu2019high} for detailed discussions on related issues.


\subsection{Inverse probability weighting}

Recall that the covariance matrix of $\by_t$, under the factor model (\ref{eq3.1}), has the following decomposition,
$
\bSigma_y= \bB\cov(\bff_t)\bB'+\bSigma_u,
$
where columns of $\bB$ are approximately equal to the eigenvectors of $\bSigma_y$ corresponding to the first $r$ eigenvalues.  As such,  let $\widehat\bSigma_y$ be an input matrix, serving as an estimator for $\bSigma_y$. Then as described in Section \ref{sec2.2}, we can estimate the space  spanned by $\bB$ using the leading eigenvectors of $\widehat\bSigma_y$.

In the presence of missing data with exogenous missing, let  $x_{it}=1\{y_{it} \text{ is observed}\}$ and we only observe $y_{it}x_{it}$ for all $(i,t)$, in which unobserved data is set to zero. Suppose for now $w_i:=P(x_{it}=1)$ is known.  We can construct an unbiased estimator $\widehat\bSigma_y=(\widehat\sigma_{ij})$ with
$$
\widehat\sigma_{ij}:=\frac{1}{w_iw_jT}\sum_{t=1}^Ty_{it}y_{jt}x_{it}x_{jt}.
$$
In the matrix form, let  $\bY$ and $\bX$ be the  $N\times T$ matrices  of  $y_{it}$ and $x_{it}$. So we only observe $\bY\circ \bX$, where $\circ$ represents the element-wise matrix product, the Hadamard product.  Also let $\bW $ be the diagonal matrix with $w_i$ being its $i$ th diagonal entry. Then
$$\widehat\bSigma_y= \frac{1}{T} \bZ\bZ',\quad \bZ:=\bW^{-1}\bY\circ \bX.
 $$
Therefore,  columns of the loading matrix estimator $\widehat\bB$ equal to  $\sqrt{N}$ times the   top right singular vectors of $\bZ$. This method simply replaces the missing entries of $\bY$ by zero, and apply the inverse probability weighting (IPW) before applying PCA.  The IPW   has been popularly used
 in the  causal inference literature (e.g., \cite{imbens2015causal}).  Here the same idea is applied to create an unbiased estimator for the covariance matrix.

 In practice, we shall replace $w_i$ by its consistent estimators, such as $\widehat w_i:=\frac{1}{T}\sum_{t=1}^Tx_{it}$.  But  in the case of homogeneous missing, that is, $w_1=\cdots=w_N$, the IPW is not needed, because
 $\bW$ equals the identity matrix up to a constant, which does not affect the PCA on $\bY\circ \bX$.
In addition, factors can be further estimated using least squares by regressing $y_{it}x_{it}$ on the estimated loadings.

Theoretical properties were studied by \cite{abbe2017entrywise, su2019factor} under the assumption of homogenous missing. \cite{su2019factor} used this estimator as their initial value for the EM algorithm.  \cite{xiong2019large} allowed heterogenous missing and proved that the estimators are also asymptotically normal (they estimated $w_iw_j$ directly by $\frac{1}{T}\sum_{t=1}^Tx_{it}x_{jt}$).  We can also quickly derive the rate of convergence by applying Theorem \ref{th1.1}.   However, the IPW  is the least  efficient approach among all the methods to be discussed in this section. We shall verify this in a simulation study in Section \ref{sec:5.sim}.

\subsection{Regularized   matrix completion}
Regularized   matrix completion  is a powerful technique to recover missing entries from low-rank matrices.
This approach   is also much faster than the EM algorithm in handling large panels. Due to these nice properties, it has also attracted much attention in the recent econometrics literature,  e.g.,
 \cite{athey2018matrix, bai2017principal,moon2018nuclear,giglio2019thousands}.

 In the matrix form $\bY=\bM+\bU$, the goal is to recover the factor component $\bM=\bB\bF'$ when $\bY$ has missing elements.    The nuclear-norm regularization is directly applicable:
 \begin{equation}\label{eq5.1}
 \widehat\bM:=\arg\min_{\bM}\|(\bY-\bM)\circ\bX\|_F^2+\lambda\|\bM\|_n
 \end{equation}
 with tuning parameter $\lambda$.  The factors and loadings can be estimated by taking the singular vectors of $\widehat\bM$.    \cite{negahban2011estimation} and \cite{koltchinskii2011nuclear} derived the rate of convergence under the \textit{Frobenius norm}. Under suitable conditions (e.g., missing at random, restricted strong convexity, sufficiently large noise) it can be proved that
 $$
 \frac{1}{NT}\|\widehat\bM-\bM\|_F^2=O_P\left(\frac{1}{T}+\frac{1}{N}\right).
 $$
 \cite{chen2020noisy} certifies further that the convex optimization (\ref{eq5.1}) is optimal for all noise levels under Frobenius norm, operator norm, and elementwise-infinity norm. The proof is based on a novel technical device that bridges the convex optimization with a nonconvex optimization problem.
 However, this estimator is not  asymptotically normal due to the presence of shrinkage bias, so is not suitable for statistical inferences.

 \subsection{Debiased estimators}

 Several recent progress in this literature focuses on debiasing the regularized regression in order to have valid confidence intervals, e.g., \cite{chen2019inference,xia2019statistical,
 chernozhukov2019inference}. When the missing is homogeneous, $P(x_{it}=1)=p$ for all $(i,t)$, \cite{chen2019inference} proposed the following simple debiased estimator
\begin{equation} \label{eqjf04}
  \widehat{\bM}^d = H_R (\widehat{\bM} + \widehat{p}^{-1}(\bY-\widehat{\bM})\circ \bX),
\end{equation}
where $H_R(\cdot)$ is the best rank $R$ approximation in (\ref{eqjf05}), $\widehat{\bM}$ is given by (\ref{eq5.1}), $\widehat{p}$ is the sample proportion of missing data.  The idea is very intuitive.  Ignoring the weak-dependence between $\widehat{\bM}$ and $\bX$ and estimating error in $\widehat p$, we have
$$
\E (\widehat{\bM} + \widehat{p}^{-1}(\bY-\widehat{\bM})\circ \bX) \approx  \E \widehat{\bM} + \E (\bY-\widehat{\bM}) = \bM,
$$
which is approximately unbiased.  However, the estimator $\widehat{\bM} + \widehat{p}^{-1}(\bY-\widehat{\bM}\circ \bX)$ is no longer of rank $R$, which increases the variances.  This leads to use the projection as in (\ref{eqjf04}), which is asymptotically efficient in terms of both rate and pre-constant.

Alternatively,  the  debiasing can be achieved through the iterative least squares \citep{chernozhukov2019inference}.  Suppose the true number of factors, $r$, is known.

\begin{algo}\label{ag5.1}
 Debias using iterative least squares.
\begin{description}
\item[Step 1.]   Obtain $\widehat\bM$ as in (\ref{eq5.1}).

\item[Step 2.]  Let the columns of $\frac{1}{\sqrt{N}}\widehat\bB$ be the left singular vectors of $\widehat\bM$, corresponding to the first $r$ singular values.

\item[Step 3.]    Estimate the latent factors at time $t$ by 
$
\widetilde\bff_t:=\left(\sum_{i=1}^N\widehat\bb_i\widehat\bb_i'x_{it}\right)^{-1}
\sum_{i=1}^N\widehat\bb_iy_{it}x_{it}
$
and let  $\widetilde\bF=(\widetilde\bff_1,\cdots,\widetilde\bff_T)'$.

\item[Step 4.]   Update   loading estimates  by $\widetilde\bB=(\widetilde\bb_1,\cdots,\widetilde\bb_N)'$, where
 $$
 \widetilde\bb_i:=\left(\sum_{t=1}^T\widetilde\bff_t\widetilde\bff_t'x_{it}\right)^{-1}\sum_{t=1}^T\widetilde\bff_ty_{it}x_{it}.
 $$

\item[Step 5.]  The \textit{asymptotically unbiased} estimator for $\bM$ is $\widetilde\bM:=\widetilde\bB\widetilde\bF'.$

 \end{description}

 \end{algo}

   A key technical argument is  to ensure that the estimation error in $\widehat\bB$ (step 2) has no impact on the factor estimator (step 3); this is achieved by \cite{chen2019inference} using an ``auxiliary leave-one-out" argument. 

 When the missing probability $P(x_{it}=1)$ varies across $i$, there are two ways to revise the previous algorithm to achieve the asymptotic normality.  One way is to replace (\ref{eq5.1}) with a weighted regularization:
  \begin{equation}\label{eq5.2}
\min_{\bM}\|(\widehat \bW^{-1/2}\bY-\widehat \bW^{-1/2}\bM)\circ\bX\|_F^2+\lambda\|\bM\|_n,
 \end{equation}
 where $\widehat\bW$ is a diagonal matrix, whose $i$ th diagonal entry equals $\widehat w_i:=\frac{1}{T}\sum_{t=1}^Tx_{it}$.  This debiases the least squares part of the loss function, adopting the same idea of inverse probability weighting.  The remaining steps of Algorithm \ref{ag5.1} are the same.  Then the same ``auxiliary leave-one-out" technical argument of \cite{chen2019inference} still goes through.
 The other way is to apply ``sample splitting", which evenly split the columns of $\bY$ into two parts: on one part we run the penalized regression as in (\ref{eq5.1}) and obtain $\widehat\bB$, on the other part we run iterative least squares. Then exchange the two parts and re-do the estimations. The final estimator is taken as the average of the two. Suppose $u_{it}$ is serially independent, the sample splitting then artificially creates independences   among various statistics  from the  splitting sample.  See \cite{chernozhukov2019inference} for detailed descriptions of this approach.

\subsection{Block-rearrangements}


In an attempt to handle endogenous missing, \cite{bai2019matrix} proposed a block-rearrangement method. At the cost of this generality, they require that the data matrix $\bY$ should have a sufficiently large balanced sub-block after elementary rearrangements.  See \cite{cai2016structured,fan2019structured} for   related ideas.

Specifically, a preliminary step of their estimation is to rearrange the data in a shape that all the factor loadings can be estimated in one sub-block and all the factors can be estimated in another sub-block. The following example is adapted from \cite{bai2019matrix}, which gives a good illustration on this manipulation: example of the $N\times T$ matrix for $y_{it}$:
\[\begin{bmatrix}
\by_{11} & \by_{12} & \by_{13} & \by_{14} & \by_{15}\\
\by_{21} & \by_{22}^* & \by_{23} & \by_{24} & \by_{25}\\
\by_{31} & \by_{32} & \by_{33} & \by_{34}^* & \by_{35}\\
\by_{41}^* & \by_{42} & \by_{43} & \by_{44} & \by_{45}\\
\by_{51} & \by_{52} & \by_{53} & \by_{54} & \by_{55}\end{bmatrix} \Longrightarrow \begin{bmatrix}
 \by_{13} & \by_{15} & {\color{blue}\by_{11}} & {\color{blue}\by_{12}} & {\color{blue}\by_{14}}\\
 \by_{53} & \by_{55} & {\color{blue}\by_{51}}  & {\color{blue}\by_{52}} & {\color{blue}\by_{54}}\\
 {\color{red}\by_{23}} & {\color{red}\by_{25}} & {\color{cyan}\by_{21}} & {\color{cyan}\by_{22}^*} & {\color{cyan}\by_{24}}\\
 {\color{red}\by_{33}} & {\color{red}\by_{35}} & {\color{cyan}\by_{31}} & {\color{cyan}\by_{32}} & {\color{cyan}\by_{34}^*}\\
 {\color{red}\by_{43}} & {\color{red}\by_{45}} & {\color{cyan}\by_{41}^*} & {\color{cyan}\by_{42}} & {\color{cyan}\by_{44}}\end{bmatrix}.\]
The left matrix is the originally collected data and the right is the rearranged one. The symbols with asterisk denote the missing data. From the column perspective, the 1st, 2nd and 4th columns have missing values and therefore are rearranged as the last three columns in the right panel; from the row perspective, the 2nd, 3rd and 4th rows have missing values and therefore are rearranged as the last three rows in the right panel.
  \cite{bai2019matrix}  name the black block ``bal'', name the black plus the red blocks ``tall'', and name the black plus the blue block ``wide''.

  Consider the missing value $\by_{22}^*$. We want to replace it with its expected value $\mathbb E(\by_{22}^*)=\bb_2'\bff_2$. Note that $\by_{22}^*$ shares the same factor loadings $\bb_2$ with data points $\by_{23}$ and $\by_{25}$ in the wide block; and shares the same factors $\bff_2$ with data points $\by_{12}$ and $\by_{52}$ in the tall block. Meanwhile, $\bb_2$ can be estimated using data in the``tall" block; $\bff_2$ can be estimated using data in the ``wide" block.   As a result, one might expect to recover $\mathbb E(\by_{22}^*)$ with these two estimators. However, we must take into account the rotational indeterminacy inherent with the factor models. 
  For a generic missing value $\by_{it}$,
\[\widehat\bb_{\mathrm{tall},i}=\bH_{\mathrm{tall}}'\bb_i+o_P(1), \quad \widehat\bff_{\mathrm{wide},t}=\bH_{\mathrm{wide}}^{-1}\bff_t+o_P(1).\]
Therefore
\[\bb_i'\bff_t=\widehat\bb_{\mathrm{tall},i}' \bA\widehat\bff_{\mathrm{wide},t}+o_P(1),\quad \bA:=  \bH_{\mathrm{tall}}^{-1}\bH_{\mathrm{wide}}  .\]
To estimate    $\bA$, by $\widehat\bff_{\mathrm{wide},t}=\bH_{\mathrm{wide}}^{-1}\bff_t+o_P(1)$ and $\widehat\bff_{\mathrm{tall},t}=\bH_{\mathrm{tall}}^{-1}\bff_t+o_P(1)$, we have $$\widehat\bff_{\mathrm{tall},t}=\bA\widehat\bff_{\mathrm{wide},t}+o_P(1).$$ So one can run the regression of $\widehat\bff_{\mathrm{tall},t}$ on $\widehat\bff_{\mathrm{wide},t}$ to  consistently estimate $\bA$. This leads to the following estimation procedure.
\begin{algo}
Block-rearrangement algorithm
\begin{description}
\item[Step 1] Obtain estimators  $(\widehat\bb_{\mathrm{wide}}, \widehat\bF_{\mathrm{wide}})$  using  the tall block of $\bY$.
\item[Step 2] Obtain estimators     $(\widehat\bb_{\mathrm{tall}}, \widehat\bF_{\mathrm{tall}})$ using the wide block of $\bY$.
\item[Step 3] Compute $\widehat \bC_{\mathrm{miss}}=\widehat\bB_{\mathrm{tall}}\bA\widehat\bF_{\mathrm{wide}}'$ where $\bA$ is obtained by regressing $\widehat\bff_{\mathrm{tall},t}$ on $\widehat\bff_{\mathrm{wide},t}$
\item[Step 4] Output $\widetilde \bY$, where $\widetilde y_{it}=y_{it}$ if  $y_{it}$ is observable;  $\widetilde y_{it}=\widehat c_{\mathrm{miss},it}$ if  $y_{it}$ is missing.
\end{description}
\end{algo}

Once $\widetilde\bY$ is obtained, we apply the PCA again to the imputed data $\widetilde\bY$ to get more efficient estimates of $\bB$ and $\bF$. Suppose the size of the ``tall'' block is $N\times T_0$ and the size of the ``wide'' block is $N_0\times T$. So the size of the ``bal'' block is $N_0\times T_0$. The whole sample size (including missing data points) is $N\times T$. Bai and Ng require that
$$
 \max\{\sqrt{N}, \sqrt{T}\} =o(N_0), \quad \text{ and } \quad  \max\{\sqrt{N}, \sqrt{T}\} =o(T_0).
$$
An implication of the above condition is that the missing data points should not be too frequent in the sense that the balanced subblock is large enough. Though this condition rules out the case of random missing (e.g., missing occurs as outcomes of Bernoulli trials), it is not stringent given the nature of endogenous missing.   

\subsection{A simulation study}\label{sec:5.sim}

We conduct a simulation study to compare six matrix completion approaches, namely:

\textbf{IPW.} The  inverse probability weighting.

\textbf{ReUW.} Unweighted regularization. The eigenvectors of the  estimator (\ref{eq5.1}).

\textbf{ReW.} Weighted regularization. The eigenvectors of the  estimator (\ref{eq5.2}).


\textbf{ReDebias.} The debiased regularized estimator  from Algorithm \ref{ag5.1}.

\textbf{EM.} The EM algorithm.

We generate a two-factor model where loadings, factors and $u_{it}$ are   independent standard normal. Under the homogeneous missing we generate  $x_{it}\sim$ Bernoulli$(0.5)$; under the heterogeneous missing we generate  $x_{it}|w_i\sim$ Bernoulli$(w_i)$, and $w_i\sim $Uniform$[0.1,1]$.  The three regularized methods require choosing $\lambda$, the tuning parameter. Write the penalized loss function to be $\|( \bW^{-1/2}\bY- \bW^{-1/2}\bM)\circ\bX\|_F^2+\lambda\|\bM\|_n$ where $\bW$ is a diagonal weighting matrix.  The theory requires that with a high probability, there is $c>0$,
$$
 (2+c)\|\bU\circ(\bW^{-1}\bX)\|<\lambda.
$$
So we set $\lambda$ to be the 0.95  quantile of $ 2.2\|\bZ\circ(\bW^{-1}\bX)\|$ where $\bZ$ is an $N\times T$ matrix of standard normal variables. In practice, one can also simulate $\bZ$ using the estimated idiosyncratic covariance matrix.

\begin{table}[h]
\tabcolsep7.5pt
\caption{Comparison among five matrix completion methods}
\label{tab1}
\begin{center}
\begin{tabular}{cc|ccccc}
\hline
\hline
$N$&$T$ &IPW  &ReUW&ReW&ReDebias2&EM \\
\hline
&&&&&\\
 & &\multicolumn{5}{c}{Homogeneous missing} \\
100 &200 & 0.176  &0.116    & 0.114 &  0.109&  0.109\\
200 &100  & 0.252 &0.171 & 0.169 & 0.161& 0.161\\
&&&&&\\
 & &\multicolumn{5}{c}{Heterogeneous missing} \\
100 &200 &0.263 &0.211& 0.131 &0.119& 0.119\\
200 &100  &0.369   &0.304  & 0.222 & 0.204 &0.203\\
\hline
\end{tabular}
\end{center}
\begin{tabnote}
Reported is $\|\bP_{\widehat\bB}-\bP_\bB\|$ averaged over 100 replications.
\end{tabnote}
\end{table}

 We compare the performance of estimating the loading space, measured by $\bP_{\bB}=\bB(\bB'\bB)^{-1}\bB'$. Table \ref{tab1} reports  $\|\bP_{\widehat\bB}-\bP_\bB\|$ averaged over 100 replications for each method.   In all scenarios, the IPW performs the worst among all estimators.  Under the homogeneous missing, all the other four methods perform similarly, but the difference is much more noticeable under the heterogeneous missing. The general ranking is that
  $$
\text{IPW}  \prec  \text{ReUW}  \prec \text{ReW} \prec\text{ReDebias} \approx \text{EM}.
$$
This ranking is as expected: IPW is the least efficient method among the five;  ReUW uses the nuclear-norm regularized estimation that does not take into account the heterogeneous missing or debias; ReW accounts for the heterogeneous missing probabilities, and ReDebias  further removes the regularization bias. 

 Finally, it is not surprising to see that ReDebias and EM perform similarly because both start with an initial low-rank estimator (ReDebias initializes from ReW while EM initializes from IPW), then proceed via iterative least squares. But we note that  ReDebias  operates much faster because it only iterates once, so is more attractive than EM in handling large scale problems. 
 We also implemented the ``early-stop-EM" (which only iterates twice), it performs only slightly better than IPW and is worse than all the other estimators. Therefore we conclude that the ReDebias is a recommended method for handling large scale low-rank matrix completion problems.

\section{Conclusion}

We have conducted  a selective overview on the recent developments of the factor model and its application on statistical learning.  We focus on the perspective of the low-rank structure of factor models,
and particularly draws attentions to estimating the model from the low-rank recovery point of view. New estimation and inference methods, and matrix completion problems have been discussed.  




\appendix

\section{Technical details}

\subsection{Proof of Theorem \ref{th1.1}}\label{sec:prof2.1}

\begin{proof} (i) The proof is an exercise of applying the eigen-perturbation theorem.
First, by  the triangular inequality, for $a_N:=2\eta_N  \|\bL\|+\eta_N^2+3\|\bL\|\|\bS\|$,
$$
\| \widehat\bSigma\widehat\bSigma' -\bL\bL'\|
\leq \|\bSigma\bSigma' -\widehat\bSigma\widehat\bSigma'\|
+\|\bSigma\bSigma'-\bL\bL'\|\leq O_P(a_N).
$$
 So by Weyl's theorem (cited in Theorem \ref{tha.2}),    $\max_{i\leq r+1}|\lambda_i^2  (\widehat\bSigma)-\lambda^2_i(\bL)|\leq O_P(a_N).$

Also,
$g_N^2:= \min_{2\leq i\leq r+1}| \lambda_{i-1}(\bL)-\lambda_i(\bL)|^2
 \gg a_N $ implies  $\lambda_r^2(\bL)\gg a_N$ and
 $$\min_{2\leq i\leq r+1}| \lambda_{i-1}^2(\bL)-\lambda^2_i(\bL)|
 \geq 2\lambda_r(\bL)\min_{2\leq i\leq r+1}| \lambda_{i-1}(\bL)-\lambda_i(\bL)|
 \gg a_N.$$
  So with probability approaching one,
\begin{eqnarray}\label{eqa2.3}
\max_{i\leq r}|\lambda_i (\widehat\bSigma)-\lambda_i(\bL)|&\leq&
\frac{\max_{i\leq r}|\lambda_i^2  (\widehat\bSigma)-\lambda^2_i(\bL)|}{\min_{i\leq r}\sqrt{2\lambda_i^2(\bL)-| \lambda_i^2  (\widehat\bSigma)-\lambda^2_i(\bL)  |}} =O_P\left(\frac{a_N}{\lambda_r(\bL)}\right).\\
\min_{i\leq r} | \lambda_{i-1}^2(\widehat\bSigma)-\lambda^2_i(\bL)|&\geq& \min_{2\leq i\leq r}| \lambda_{i-1}^2(\bL)-\lambda^2_i(\bL)|
  -| \lambda_{i-1}^2(\widehat\bSigma)-\lambda^2_{i-1}(\bL)|\cr
  &\geq& \frac{1}{2}\min_{2\leq i\leq r}| \lambda_{i-1}^2(\bL)-\lambda^2_i(\bL)|.
\end{eqnarray}

Similarly, for all $i\leq r-1$, $|\lambda^2_i(\bL)- \lambda^2_{i+1}(\widehat\bSigma)|\geq \frac{1}{2}\min_{2\leq i\leq r}| \lambda_{i-1}^2(\bL)-\lambda^2_i(\bL)|.$ Now for $i=r$,
 $\lambda_{r+1}^2(\widehat\bSigma)\leq \lambda^2_{r+1}(\bL)+ O_P(a_N)=O_P(a_N)$.
 So $|\lambda^2_r(\bL)- \lambda^2_{r+1}(\widehat\bSigma)|\geq \frac{1}{2}\lambda^2_r(\bL)$.
  Hence by the sing-theta theorem (cited in Theorem \ref{tha.2}),
  \begin{equation}\label{eqa2.4}
 \max_{i\leq r}\|\widehat\bxi_i-\bxi_i\|
 =O_P\left(\frac{a_N}{  g^2_N}\right).
\end{equation}
The right singular vectors have the same bound. Then (\ref{eqa2.3}) (\ref{eqa2.4}) together imply
$$
\|\widehat\bL-\bL\|=O_P\left(\frac{a_N\|\bL\|}{ g^2_N}\right) .
$$
 Finally, we note that $\|\bL\|= \sum_{i=2}^{r+1}\lambda_{i-1}(\bL)-\lambda_i(\bL)\leq rg_N$. So  $a_N=o_P(g_N^2)$ is satisfied as long as $\|\bS\|+\eta_N=o_P(g_N)$ and $a_N=O_P(g_N\eta_N+g_N\|S\|)$.

 (ii) The element-wise bound  is a corollary from the more general bound in Theorem \ref{eqtha.1}.
Using the notation of  Theorem \ref{eqtha.1}, under the assumptions that $\|\bSigma\|_\infty=O_P(1)$, $s_N=O_P(\sqrt{N_1})$. Also, $m_N=O_P(\frac{1}{\sqrt{N}}+\frac{1}{\sqrt{N_1}})$, $\|\bL\|_\infty = O_P(1) $, $N_1c_N=o_P(g_N)$.
Hence
  $$
  b_N\leq \left(   \frac{N_1}{\sqrt{N}}+  \sqrt{N_1}   \right)  g_N^{-2}  (\eta_N+\|\bS\|)
+\left(c_N \frac{ N_1}{\sqrt{N}}   +  c_N  \sqrt{N_1}   +\|\bS\bzeta_d \|_\infty \right) g_N^{-1}.
  $$

 \end{proof}

\subsection{Proof of Theorem \ref{th2.1}}\label{sec:3.1profd}

\begin{proof}   Let $Q(\bL, \bSigma_u)$ denote the loss function. Note that
	\begin{eqnarray*}
		\|\bS_y-\widehat\bL-\widehat\bSigma_u\|_F^2&=&\|\bS_y-\bSigma_y\|_F^2
		+\|\widehat \bL+\widehat \bSigma_u-\bL-\bSigma_u\|_F^2+
		M_1+M_2\cr
		M_1&=&2\tr((\bS_y-\bSigma_y)(\bL-\widehat \bL)')\cr
		M_2&=&2\tr((\bS_y-\bSigma_y)(\bSigma_u-\widehat\bSigma_u)').
	\end{eqnarray*}
	
	We now use two types of  inequalities to bound $M_1$ and $M_2$. As for $M_1$, note that $\bL$ is a low-rank matrix, we use the inequality
	$|\tr(\bA\bB')|\leq \|\bA\|\|\bB\|_n$. We thus have
	$$
	|M_1|\leq 2\|\bS_y-\bSigma_y\|\|\bL-\widehat \bL\|_n\leq 0.5\nu_1\|\bL-\widehat \bL\|_n.
	$$
	As for $M_2$, note that $\bSigma_u$ is a sparse matrix, we use the inequality $|\tr(\bA\bB')|\leq \|\bA\|_\infty\|\bB\|_1$,
	$$
	|M_2|\leq 2\|\bS_y-\bSigma_y\|_\infty \|\bSigma_u-\widehat\bSigma_u\|_1\leq 0.5\nu_2\|\bSigma_u-\widehat\bSigma_u\|_1.
	$$
	From Lemma 2.3 of \cite{recht2010guaranteed}, $\|\bL+\mathcal P(\bA_1)\|_n=\|\bL\|_n+\|\mathcal P(\bA_1)\|_n$ where  $\bA_1=\widehat\bL-\bL$. Also using the standard sparse argument, $\|\bSigma_u+(\bA_2)_{J^c}\|_1=\|\bSigma_u\|_1+\|(\bA_2)_{J^c}\|_1$ where  $\bA_2=\widehat\bSigma_u-\bSigma_u$.
	In addition, Lemma 1 of \cite{negahban2011estimation} shows $0.5\text{rank}(\mathcal M(\bA_1))\leq r:=\text{rank}(\bL)$.
	Hence
	\begin{eqnarray*}
		\|\widehat \bL\|_n&=&\|\bL+\mathcal P(\bA_1)+\mathcal M(\bA_1)\|_n
		\geq \|\bL\|_n+\|\mathcal P(\bA_1)\|_n-\|\mathcal M(\bA_1)\|_n\cr
		\|\widehat\bSigma_u\|_1&=&\|\bSigma_u+(\bA_2)_{J^c}+(\bA_2)_J\|_1
		\geq \|\bSigma_u\|_1+\|(\bA_2)_{J^c}\|_1-\|(\bA_2)_J\|_1,\cr
		\|\mathcal M(\bA_1)\|_n&\leq& \|\mathcal M(\bA_1)\|_F\sqrt{\text{rank}(\mathcal M(\bA_1))}
		\leq \sqrt{2r}\|\bA_1\|_F\cr
		\|(\bA_2)_J\|_1&\leq&\|\bA_2\|_F\sqrt{J+N}.
	\end{eqnarray*}
	Thus  $Q(\widehat\bL,\widehat\bSigma_u)\leq Q(\bL, \bSigma_u)$ implies
	\begin{eqnarray*}
		&&\|\bA_1+\bA_2\|_F^2+0.5\nu_1\|\mathcal P(\bA_1)\|_n
		+0.5\nu_2 \|(\bA_2)_{J^c}\|_{1}
		\leq
		1.5\nu_1\|\mathcal M(\bA_1)\|_n
		+1.5\nu_2\|(\bA_2)_{J}\|_1.
	\end{eqnarray*}
	As such, $(\bA_1, \bA_2)\in\mathcal C(\nu_1, \nu_2)$, and thus $\|\bA_1+\bA_2\|_F^2\geq\kappa(\nu_1, \nu_2)(\|\bA_1\|_F^2+\|\bA_2\|_F^2) $.
	\begin{eqnarray*}
		\kappa(\nu_1, \nu_2)(\|\bA_1\|_F^2+\|\bA_2\|_F^2)
		&\leq& 1.5\nu_1\|\mathcal M(\bA_1)\|_n
		+1.5\nu_2  \|  (\bA_2)_{J} \|_1\cr
		&\leq& \sqrt{4.5r}\nu_1\|\bA_1 \|_F
		+1.5\nu_2\sqrt{J+N}  \| \bA_2 \|_F.
	\end{eqnarray*}
	The last inequality then implies $\|\bA_1\|_F+\|\bA_2\|_F\leq \frac{4}{\kappa(\nu_1,\nu_2)} (\sqrt{4.5r}\nu_1+1.5\nu_2\sqrt{J+N} )$.
	
\end{proof}

\subsection{Proof of Theorem \ref{th2.2}}\label{sec:prof3.2}

\subsubsection{Proof of Theorem \ref{th2.2} (i)}
\begin{proof}
	First, we have $|\mathbb E\widetilde y_i-\mathbb E y_{it}|\leq \mathbb E(|y_{it}|-\tau_i)1\{|y_{it}|>\tau_i\}\leq 2\sigma\sqrt{\mathbb E|y_{it}|^q}\tau_i^{-q/2}$ for any $q\geq 2$.
	Meanwhile,  $\mathbb E|\widetilde y_{it}|^q\leq\tau_i^{q-2} \sigma^2 .$
	As such we can apply the Bernstein's inequality (Theorem \ref{tha.1})  to reach:
	$$
	\max_{i\leq N}P\left(\left|\widetilde y_i-\mathbb E\widetilde y_i\right|>\sqrt{\frac{2\sigma^2 x^2}{T}}+\frac{\tau_i x^2}{T}\right)<2\exp(-x^2).
	$$
	Then take $x=\sqrt{4\log N}$, by union bound, with probability at least $1-2N^{-3}$,
	$$
	\max_{i\leq N}\left|\widetilde y_i-\mathbb E\widetilde y_i\right|\leq \sqrt{\frac{8\sigma^2 \log N}{T}}+\frac{4\max_{i\leq N}\tau_i \log N}{T}.
	$$
The rest of the proof is conditioning on this event. Together we have
	$$
	\max_{i\leq N}\left|\widetilde y_i-\mathbb Ey_{it}\right|\leq \sqrt{\frac{8\sigma^2 \log N}{T}}+\max_{i\leq N} f(\tau_i) ,\quad f(\tau):=\frac{4\tau \log N}{T}+\frac{2\sigma\sqrt{   M}}{\tau^{q/2}}.
	$$
	Set $\tau_i\asymp\kappa(\sigma^2M)^{1/(2+q)} $, which is proportional to the minimizer of $f(\tau).$ Here   $\kappa:= \left(\frac{T}{\log N}\right)^{1/(1+q/2)}$. Because $\log N\leq CT$, we have $\kappa \sqrt{\frac{\log N}{T}}<C_2$ for some constant $C_2$ as long as $q\geq 2$. So
		$$
	\max_{i\leq N}\left|\widetilde y_i-\mathbb Ey_{it}\right|\leq \left(\sqrt{8} + M^{1/(2+q)} \kappa \sqrt{\frac{\log N}{T}} \right) \sigma  \sqrt{\frac{\log N}{T}}  \leq (M^{1/(2+q)}C_2+\sqrt{8}) \sigma  \sqrt{\frac{\log N}{T}} .
	$$

\end{proof}
\subsubsection{Proof of Theorem \ref{th2.2} (ii)}

\begin{proof}
Let $\mu_{i,\tau}:= \arg\min_{\mu} \E\psi_{\tau_i}(y_{it}-\mu)$.
We separately consider the ``variance" $\widehat y_i- \mu_{i,\tau}$ and the ``bias"  $\Delta_{i}:=\mu_{i,\tau}-\E y_{it}$. We denote by
$$\dot\psi_{\tau}(x):=\frac{d}{dx}\psi_{\tau}(x)
=\begin{cases} 2x\tau^{-2}, & |x|<\tau\\
2\sgn(x)\tau^{-1}, & |x|\geq \tau,
\end{cases}
$$
which is well defined because $\psi_{\tau}$ is first-order differentiable (but not twice).

\textbf{bias.} Bounding bias requires $\tau_i$ cannot grow too slowly.  Let $g_{\tau}(z):=z^2-\tau^2\psi_\tau(z)$ and $\frac{d}{dz}g_\tau(z)=2z-\tau^2\dot\psi_\tau(z)$.
Hence for $z_1=y_{it}-\mu_{i,\tau}$ and $z_2=y_{it}-\mathbb E y_{it}$. There is  $\widetilde \mu$ in between $\mu_{i,\tau}$ and $\E y_{it}$, for $\widetilde z=y_{it}-\widetilde \mu$, and $q>1$,
\begin{eqnarray*}
\E (z_1^2-z_2^2)&=&\E   g_{\tau_i}(z_1)-\E  g_{\tau_i}(z_2)+\tau_i^2\E  \psi_{\tau_i}(z_1)-\tau_i^2\E  \psi_{\tau_i}(z_2)
\leq\E   g_{\tau_i}(z_1)-\E  g_{\tau_i}(z_2)\cr
&=&\E \frac{d}{dz}g_\tau(\widetilde z)(\E y_{it}-\mu_{i,\tau})
\leq2 |   \Delta_i|\E 1\{|\widetilde z|\geq\tau\}  |\widetilde z|\leq 2|  \Delta_i   | \E  |\widetilde z|^{q}\tau_i^{-(q-1)}\cr
&\leq & 2|   \Delta_i| \E  |e_{it}+\mathbb E y_{it}-\widetilde \mu|^{q}\tau_i^{-(q-1)}\leq \tau_i^{-(q-1)} C|  \Delta_i|
[ \E  |e_{it}|^q+ |  \Delta_i|^{q}].
\end{eqnarray*}
On the other hand, $\E( z_1^2-z_2^2)= \Delta_i^2$. Hence $|\Delta_i|\leq \tau_i^{-(q-1)} C
[ \E  |e_{it}|^q+ |  \Delta_i|^{q}]$. Now consider two cases.

Case 1: $|\Delta_i|^q> \E  |e_{it}|^q$, then
$\E( z_1^2-z_2^2)= \Delta_i^2$. Hence $|\Delta_i|\leq 2 C\tau_i^{-(q-1)}  |  \Delta_i|^{q}$. This implies
 $ \tau_i \leq  C |  \Delta_i| \leq C|\mu_{i,\tau}+\E y_{it}|$ which contradicts with $\tau_i\to \infty.$

Case 2:  $|\Delta_i|^q< \E  |e_{it}|^q$, then $|\Delta_i|\leq C\tau_i^{-(q-1)}
  \E  |e_{it}|^q\to0 $ if $ \E  |e_{it}|^q<\infty.$

\textbf{Variance.} Bounding variance requires $\tau_i$ cannot grow too fast.  Denote the loss function
$Q_i(\mu):=\frac{1}{T}\sum_{t=1}^T\psi_{\tau_i}(y_{it}-\mu)$. Fix $m_T=\sqrt{\frac{\log N}{T}}$. We aim to show   there is $\delta>0$,  so that
\begin{equation}\label{eqa.1}
P\left(\inf_{|\nu|=\delta}\min_{i\leq N}Q_i(\mu_{i,\tau}+m_T\nu)-Q_i(\mu_{i,\tau})>0\right)>1-4N^{-3},
\end{equation}
which then implies with probability at least $1- 4N^{-3}$, $\max_{i\leq N}|\widehat y_i-\mu_{i,\tau}|\leq m_T\delta$. To prove (\ref{eqa.1}), note
$\E \dot\psi_{\tau_i}(e_{it,\tau})=0$ where $e_{it,\tau}:=y_{it}-   \mu_{i,\tau}.$ Now let $e_{it}=y_{it}-\E y_{it}$, then $e_{it}=  e_{it,\tau}+\Delta_{i}$, where $\Delta_{i}=\mu_{i,\tau}-\E y_{it}$.
It can be verified that for any $x, x_1, x_2$,
\begin{eqnarray*}
\dot\psi_{\tau_i}(e_{it,\tau}+x)-\dot\psi_{\tau_i}(e_{it,\tau})&=&2x\tau_i^{-2}+ a_{it}(x)b_{it}(x) \cr
\psi_{\tau}(e_{it,\tau}+x)- \psi_{\tau}(e_{it,\tau})&=&\dot\psi_{\tau}(e_{it,\tau})x+\int_0^x[\dot\psi_{\tau}(e_{it,\tau}+z)-\dot\psi_{\tau}(e_{it,\tau})]dz\cr
|\dot\psi_{\tau}(x_1)-\dot\psi_{\tau}(x_2)|&\leq& 2\tau^{-2}|x_1-x_2|,
\end{eqnarray*}
where $a_{it}(x)=\dot\psi_{\tau_i}(e_{it,\tau}+x)-\dot\psi_{\tau_i}(e_{it,\tau})-2x\tau_i^{-2} $ and $b_{it}(x)=1\{ |e_{it,\tau}+x|\vee |e_{it,\tau}|\geq \tau_i\}$;  $a\vee b=\max\{a,b\}.$ Also, $|a_{it}(x)|\leq 4|x|\tau_i^{-2}$.
Applying  these results with $x= -m_T\nu$, we have
\begin{eqnarray*}
&&Q_i(\mu_{i,\tau}+m_T\nu)-Q_i(\mu_{i,\tau})=\frac{1}{T}\sum_{t=1}^T\psi_{\tau_i}(e_{it,\tau}+x)
-\psi_{\tau_i}(e_{it,\tau})\cr
&=&\frac{1}{T}\sum_{t=1}^T \dot\psi_{\tau_i}(e_{it,\tau})x
+\frac{1}{T}\sum_{t=1}^T\int_0^x 2z\tau_i^{-2} dz
+\frac{1}{T}\sum_{t=1}^T\int_0^x a_{it}(z)b_{it}(z)dz1\{x>0\}
-\frac{1}{T}\sum_{t=1}^T\int_x^0 a_{it}(z)b_{it}(z)dz1\{x\leq 0\}\cr
&\geq&   x^2\tau_i^{-2}
-x\tau_i^{-2}\cdot\underbrace{\tau_i^{2}\max_{i\leq N} \left|\frac{1}{T}\sum_{t=1}^T \dot\psi_{\tau_i}(e_{it,\tau})\right|}_{I}
-\tau_i^{-2}\underbrace{ \max_{i\leq N} 8 \frac{1}{T}\sum_{t=1}^T\int_0^{|x|}  z b_{it}(z)dz}_{II}.
\end{eqnarray*}
To bound $I$, we apply the Bernstein inequality.  $| \dot \psi_{\tau_i}(e_{it,\tau})|\leq 2\min\{(|e_{it}|+|\Delta_{i}|)\tau_i^{-2},\tau_i^{-1}\} $. Hence
$\E  \psi_{\tau_i}(e_{it,\tau})^2\leq 8 \tau_i^{-4} ( \sigma^2+1)$ and
$\E | \psi_{\tau_i}(e_{it,\tau})|^k\leq 8 \tau_i^{-4} ( \sigma^2+1)(\frac{2}{\tau_i})^{k-2}$
 where we used $|\Delta_{i}|<1$, and $\E|x|^k\leq \E x^2 C^{k-2}$ if $|x|<C.$ Hence by
Theorem \ref{tha.1},
$$
	\max_iP\left(\left|\frac{1}{T}\sum_{t=1}^T  \dot\psi_{\tau_i}(e_{it,\tau})\right|>\sqrt{\frac{16 \tau_i^{-4} ( \sigma^2+1) h^2}{T}}+\frac{2h^2}{\tau_iT}\right)<2\exp(-h^2).
	$$
Take $h^2=4\log N$, by the union bound, with probability at least $1-2N^{-3}$,
$$
I\leq \sqrt{\frac{64   ( \sigma^2+1) \log N}{T}}+\frac{8\tau_i\log N}{ T}\leq C(\sigma+1)m_T
$$
where $C$ does not depend on $i$. The last inequality holds for $\tau_i\asymp m_T^{-1}$.

To bound $II$, note
$b_{it}(z)\leq 2\times 1\{ |e_{it}|>\tau_i/4\} +2\times 1\{|\Delta_{i}|>\tau_i/4\}+1\{|z|>\tau_i/2\}$. Also,
when $|z|\leq |x|\leq |\nu|m_T\to0$, we have  $1\{|z|>\tau_i/2\}=0$ and
$1\{|\Delta_{i}|>\tau_i/4\}=0$  because $m_T\to 0$.
We     apply Hoeffding inequality, with probability at least $1-2N^{-3}$,
\begin{eqnarray*}
II&\leq& \max_{i\leq N} 8  x^2\frac{1}{T}\sum_{t=1}^T  1\{ |e_{it}|>\tau_i/4\}
\leq \max_{i\leq N} 8 x^2\left( 2 \sqrt{\frac{\log N}{T}}+  P(   |e_{it}|>\tau_i/4 )  \right)\cr
&\leq& \max_{i\leq N} 32x^2\left(  \sqrt{\frac{\log N}{T}}+  \frac{\sigma^2}{\tau_i^2} \right)
\leq \frac{1}{4}x^2.
  \end{eqnarray*}
Together, $Q_i(\mu_{i,\tau}+m_T\nu)-Q_i(\mu_{i,\tau})
\geq \tau_i^{-2}m_T^2|\nu|(\frac{3}{4} |\nu| -C(\sigma+1)) >0.
$ This inequality holds uniformly for all   $|\nu|=4C(\sigma+1)$ and $i\leq N$. Hence with probabiliy at least $1-4N^{-3}$,
$\max_{i\leq N}|\widehat y_i-\mu_{i,\tau}|\leq m_T4C(\sigma+1).$

Combine both the bias and variance parts, for $q\geq 2$ and $\E |e_{it}|^q<C$,
$$
\max_{i\leq N}|\widehat y_i-\E y_{it}|\leq m_T4C(\sigma+1)+ C \E |e_{it}|^q m_T^{q-1}
\leq 8Cm_T(\sigma+1).
$$

\end{proof}

\subsection{Some inequalities}

The following theorem is adapted from Theorem 2.10 of \cite{boucheron2013concentration}.
\begin{thm}[Bernstein inequality]\label{tha.1}
	Let $X_1,...,X_T$ be an independent sequence with $\frac{1}{T}\sum_{t=1}^T\mathbb E X_t^2<\sigma^2$ and $\frac{1}{T}\sum_{t=1}^T\mathbb E |X_t|^q\leq \frac{q!}{2} \sigma^2c^{q-2}$ for all integers  $q>2$, with constants $(\sigma^2,c)$. Then for all $x>0$,
	$$
	P\left(\left|\frac{1}{T}\sum_{t=1}^T(X_t-\mathbb E X_t)\right|>\sqrt{\frac{2\sigma^2 x^2}{T}}+\frac{cx^2}{T}\right)<2\exp(-x^2).
	$$
	
\end{thm}

	\begin{thm}[Eigen-purturbation bounds]\label{tha.2} Let $\lambda_1\geq...\lambda_R$ and $\widehat\lambda_1\geq...\widehat\lambda_R$ respectively be the eigenvectors of $N\times N$ semi-positive definite matrices $\bA$ and $\widehat\bA$, where $R<N$. Also, let $(\bxi_1,...,\bxi_R)$ and $(\widehat\bxi_1,...,\widehat\bxi_R)$  be corresponding eigenvectors.
				Then
				
 (i) Sin-theta theorem:
 $$
\max_{i\leq R} \|\widehat\bxi_i-\bxi_i\|\leq \frac{\|\widehat\bA-\bA\|}{\min_{i\leq R}\min\{|\widehat\lambda_{i-1}-\lambda_i|, |\lambda_i-\widehat\lambda_{i+1}|\}}.
 $$

 (ii) Weyl's theorem:
 $$
 \max_{i\leq R} \|\widehat\lambda_i-\lambda_i\|\leq \|\widehat\bA-\bA\|.
 $$
\end{thm}


Next,  we prove a   general  element-wise deviation bound for singular vectors.
  We consider the model as described in Theorem \ref{th1.1}. Let $\widehat\bzeta$ and $\bzeta$ be the $N_1\times r$ matrices of  right singular vector of $\widehat\bSigma$ and $\bL$, and let $\widehat\bxi$ and $\bxi$ be the left singular vectors.
  \begin{thm}\label{eqtha.1}
Let  $c_N:=\|\widehat\bSigma-\bSigma\|_\infty$, $\eta_N:=\|\widehat\bSigma-\bSigma\|$, $s_N^2:=\max_{i\leq N}\sum_{k\leq N_1} \Sigma_{ik}^2$, $m_N= \|\bzeta \|_\infty \vee \|\bxi\|_\infty$ and $g_N:= \min_{2\leq i\leq r+1}| \lambda_{i-1}(\bL)-\lambda_i(\bL)|$.
   Suppose $N_1c_N=o_P(g_N)$.
Then  for $b_N:=(s_N +N_1\| \bL\|_\infty m_N)  g_N^{-2}  (\eta_N+\|\bS\|)
+( N_1c_N m_N+\|\bS\bzeta_d \|_\infty) g_N^{-1},$
$$
\|\widehat\bxi-\bxi\|_\infty+  \|\widehat  \bzeta- \bzeta \|_\infty
\leq O_P(b_N).
$$
  \end{thm}

  \begin{proof}  Let $\widehat\bzeta_d$ and $\bzeta_d$ be the $N_1\times 1$ vector of the  $d$ th right singular vector of $\widehat\bSigma$ and $\bL$, for some $d\leq r$. By definition, $\bxi_d=\lambda_d^{-1}(\bL)\bL\bzeta_d$ and $\widehat   \bxi_d=\lambda_d^{-1}(\widehat\bSigma)\widehat\bSigma\widehat  \bzeta_d$.  So
$\|\widehat\bxi_d-\bxi_d\|_\infty\leq I+II+III$, where
\begin{eqnarray*}
I&:=&
  \|\lambda_d^{-1}(\widehat\bSigma)\widehat\bSigma(\widehat  \bzeta_d- \bzeta_d) \|_\infty
\cr
II&:=&\|\lambda_d^{-1}(\widehat\bSigma)(\widehat\bSigma -\bL)\bzeta_d \|_\infty\cr
III&:=&\|(\lambda_d^{-1}(\widehat\bSigma)- \lambda_d^{-1}(\bL))  \bL   \bzeta_d \|_\infty.
\end{eqnarray*}
We shall use  $\|\bA\bb\|_\infty\leq\min\{ \|\bA\|_\infty\|\bb\|_\infty N_1, \max_{i\leq N}\|\bA'_i\|\|\bb\|, \|\bA\|\|\bb\|\}$
 for $\bb\in\mathbb R^{N_1}$.
Also, part (i) shows $ \|\widehat  \bzeta_d- \bzeta_d \|=O_P(\frac{ \eta_N+\|\bS\|}{  g_N})$, the same bound as the left singular vectors.
\begin{eqnarray*}
I&\leq&
N_1 \lambda_d^{-1}(\widehat\bSigma)  \|\widehat\bSigma-\bSigma\|_\infty\|\widehat  \bzeta_d- \bzeta_d \|_\infty
+\|\lambda_d^{-1}(\widehat\bSigma)\bSigma(\widehat  \bzeta_d- \bzeta_d) \|_\infty
\leq O_P(N_1g_N^{-1}c_N)\|\widehat  \bzeta_d- \bzeta_d \|_\infty
\cr
&&+O_P(g_N^{-1})  \|\widehat  \bzeta_d- \bzeta_d \|\sqrt{\max_{i\leq N}\sum_{k\leq N_1} \Sigma_{ik}^2}
=O_P(N_1g_N^{-1}c_N)\|\widehat  \bzeta_d- \bzeta_d \|_\infty+O_P(g_N^{-2}s_N (\eta_N+\|\bS\|)).\cr
II&\leq& \lambda_d^{-1}(\widehat\bSigma) N_1\|\widehat\bSigma -\bSigma\|_\infty  \|\bzeta_d \|_\infty
+\lambda_d^{-1}(\widehat\bSigma)\|\bS\bzeta_d \|_\infty
\leq O_P(g_N^{-1} N_1c_N  \|\bzeta_d \|_\infty)
+O_P(g_N^{-1} \|\bS\bzeta_d \|_\infty )\cr
III&\leq&\|\lambda_d^{-1}(\widehat\bSigma)- \lambda_d^{-1}(\bL) \|  N_1\| \bL\|_\infty\|   \bzeta_d \|_\infty
\leq O_P(g_N^{-2}(\eta_N+\|\bS\|)) N_1\| \bL\|_\infty\|   \bzeta_d \|_\infty.
\end{eqnarray*}  Together,
$
\|\widehat\bxi_d-\bxi_d\|_\infty
\leq O_P(N_1g_N^{-1}c_N)\|\widehat  \bzeta_d- \bzeta_d \|_\infty+ O_P(b_N).
$
Similarly, $ \|\widehat  \bzeta_d- \bzeta_d \|_\infty
\leq O_P(N_1g_N^{-1}c_N)      \|\widehat\bxi_d-\bxi_d\|_\infty  + O_P(b_N)    $.
Hence for $\Delta:=\|\widehat\bxi_d-\bxi_d\|_\infty+  \|\widehat  \bzeta_d- \bzeta_d \|_\infty$, we have
$\Delta \leq O_P(N_1g_N^{-1}c_N)\Delta +O_P(b_N).
 $ Because $N_1g_N^{-1}c_N=o_P(1)$, we have $\Delta= O_P(b_N).$

\end{proof}

 \bibliographystyle{ims}
  \bibliography{liaoBib}

\end{document}